\documentclass[11pt, letterpaper, oneside]{article}
\usepackage{amsfonts, amsmath, chicago, amssymb, graphicx, theorem, titling}
\usepackage[margin=1.02in]{geometry}
\usepackage[capposition=top]{floatrow}
\usepackage{subcaption}
\usepackage{setspace}
\usepackage[table]{xcolor}
\usepackage{colortbl}
\definecolor{lightgray}{gray}{0.9}

\usepackage[bottom]{footmisc}

\usepackage{booktabs}
\usepackage{bm}
\usepackage{hyperref}
\usepackage{lscape}
\usepackage{rotating}
\usepackage[all]{xy}
\hypersetup{colorlinks,%
            citecolor=black,%
            filecolor=black,%
            urlcolor=black,%
            }

\DeclareMathOperator{\argmin}{argmin }

\newtheorem{theorem}{Theorem}

\newtheorem{lemma}[theorem]{Lemma}

\newenvironment{proof}[1][Proof]{\noindent\textbf{#1.} }{\ \rule{0.5em}{0.5em}}
\newcommand{\XpX}{X^\prime X}
\newcommand{\Xpy}{X^\prime y}
\newcommand{\XpXj}[1]{X^{#1\prime} X^{#1}}
\newcommand{\Xpyj}[1]{X^{#1\prime} y^{#1}}
\newcommand{\Zpyj}[1]{Z^{#1\prime} y^{#1}}
\newcommand{\XpZj}[1]{X^{#1\prime} Z^{#1}}
\newcommand{\ZpZ}{Z^\prime Z}
\newcommand{\ZpZj}[1]{Z^{#1\prime} Z^{#1}}
\newcommand{\Ypy}{y^\prime y}
\newcommand{\Ypyj}[1]{y^{#1\prime} y^{#1}}
\newcommand{\XpWXj}[1]{X^{#1\prime} W^{#1}X^{#1}}
\newcommand{\XpWyj}[1]{X^{#1\prime} W^{#1}y^{#1}}
\newcommand{\XpZ}{X^\prime Z}
\newcommand{\Zpy}{Z^\prime y}

\usepackage{algpseudocode}
\usepackage[ruled,vlined]{algorithm2e}
\title{\vspace{-1.5cm}\hspace{2.5cm}\textbf{(Frisch-Waugh-Lovell)}$^\prime$ \newline On the Estimation of Regression Models by Row}

\author{
Damian Clarke \\ \emph{Department of Economics, University of Exeter \& University of Chile \& MIPP} \\ \href{mailto:dclarke@fen.uchile.cl}{\color{red}{dclarke@fen.uchile.cl}}\\
\and Nicolás Paris Torres \\ \emph{Department of Economics, University of Chile} \\ \href{mailto:nparis@fen.uchile.cl}{\color{red}{nparis@fen.uchile.cl}} \\ 
\and Benjamín Villena-Roldán \\ \emph{Department of Economics, Universidad Andres Bello \& LM$^2$C$^2$ \& MIPP} \\ \href{mailto:benjamin@benjaminvillena.com}{\color{red}{benjamin@benjaminvillena.com}}
}

\date{\today}

\begin{document}
\begin{spacing}{1}
\maketitle
\thispagestyle{empty}
\begin{abstract}
We demonstrate that regression models can be estimated by working independently in a row-wise fashion. We document a simple procedure which allows for a wide class of econometric estimators to be implemented cumulatively, where, in the limit, estimators can be produced without ever storing more than a single line of data in a computer's memory.  This result is useful in understanding the mechanics of many common regression models.  These procedures can be used to speed up the computation of estimates computed via OLS, IV, Ridge regression, LASSO, Elastic Net, and Non-linear models including probit and logit, with all common modes of inference.  This has implications for estimation and inference with `big data', where memory constraints may imply that working with all data at once is particularly costly. We additionally show that even with moderately sized datasets, this method can reduce computation time compared with traditional estimation routines.
\end{abstract}


\vspace{3.5cm}
\hrule
{\footnotesize \noindent\textbf{Acknowledgements:} We thank Richard Blundell, Samuel P. Engle, Sebastian Kripfganz, James MacKinnon, and Jeffrey Wooldridge for their feedback and suggestions, and are grateful to Iván Gutierrez Martínez for excellent research assistance.  The authors acknowledge the Millenium Institute for Research in Market Imperfections and Public Policy (MIPP) for financial and institutional support.  Clarke: University of Exeter, University of Chile, IZA, MIPP and CAGE.  Paris Torres: University of Chile.  Villena Roldán: Universidad Andrés Bello, MIPP, and  LM$^2$C$^2$.}
\clearpage
\end{spacing}
\begin{spacing}{1.4}
\section{Introduction}
The Frisch-Waugh-Lovell theorem is a canonical result in econometrics, and the foundation of many modern econometric estimation procedures.  That a regression can be estimated by partitioning data column-wise is intuitive, and has a multitude of applications when brought to real data.  Perhaps surprisingly, especially in a time where datasets are growing ever-larger and more decentralised, relatively little attention has been paid to the \emph{row}-wise consideration of this problem.  In this paper we seek to address the question of whether and how regression models can be estimated when partitioning data by row. 

We show that many regression models can be estimated by partitioning data in blocks of rows and that these partitions can be arbitrarily small or large.  This implies that for a large class of regression models, there is no need for data ever to be stored in a single matrix, or ever stored in a computer's working memory.  As well as the conceptual elegance of this result, we show that it can be of substantial use, especially when data is large.  When data is so large that it escapes the working memory of a computer, this row-wise partitioning of data implies that estimation can proceed, with processing time simply scaling linearly with the number of observations.  However, even where data is not too large to fit in a computer's working memory, we show that this result may offer a speed-up over standard commercial regression implementations, where in practice, processing times tend to not scale linearly with observations.

The basic intuition of this result is that many regression models require taking sums over cross-products of matrices of data such as a series of independent variables $X$ of dimension $N\times K$.  A clear example is the OLS estimator $\widehat\beta_{OLS}=(X^\prime X)^{-1}X^\prime y$.  In practice, calculating $X^\prime X$ requires summing over all observations $i\in\{1,\ldots,N\}$ the product of each value of each independent variable within observations $i$, but \emph{does not} require cross-products taken across observations $j\neq i$.  Similarly, $X^\prime y$ requires that for a given observation $i$, the value of each of $k$ independent variables $x_{ki}$ be multiplied by $y_i$, but such a cross-multiplication is not required between observations.  This implies that one can arrive very simply to aggregates such as $X^\prime X$ and $X^\prime y$ without ever reading an entire dataset into memory.  Indeed, in the limit, one can calculate these quantities by reading a line of data for a single observation $i$ at a time, iterating over all $N$, but never holding more than a single line of raw data in memory.\footnote{A simple visualisation of this calculation is provided in Appendix \ref{app:simpleVis}.}  This may seem surprising given that regression models account fully for the interdependence between independent variables, but it is a base result of matrix algebra and the mechanics of multivariate regression. 

What is more, a similar sequential procedure can be conducted for the variance of estimates of regression coefficients, as well as standard goodness-of-fit parameters such as the R-squared, implying that standard errors, confidence intervals, and any hypothesis tests can also be calculated exactly without ever loading all data in memory.  We show that a similar logic can be used for alternative models such as two-stage least squares (2SLS), penalised regression models such as LASSO, Ridge and elastic net, and the estimation of probit and logit models using indirect least squares.  Similar results can be derived for models where maximum likelihood (ML) estimation procedures are used.  The methods discussed here are surprisingly flexible, also being feasible (and comparatively fast) with estimates which \emph{prima facie} one may believe require loading larger portions of data in memory.  For example, we show that similar \emph{row}-wise procedures can be used for cluster-robust variance-covariance estimates without ever reading data on an entire group of observations at once. 

In this paper we begin by establishing that linear regression models can be estimated row-wise, without ever opening the entire dataset.  We define a ``cumulative ordinary least squares'' algorithm, which is exactly identical to OLS in both point estimates and standard errors, as well as any of the other basic statistics desired which are commonly reported following the estimation of linear models by OLS.  This result has historical precedents in early computational literature in economics; see for example \citeN{Brownetal1953} who note that a specific variant of this procedure can be used.  However, we \emph{also} document that this result extends to virtually all commonly used alternative variance estimators and regression models.  We discuss the computational implementation of such an estimator, noting that one could elect to split data into arbitrarily small or arbitrarily large partitions, though in practice, partitions should be sufficiently small such that they do not approach the limits of a computer's working memory.

In our research, we highlight an intriguing aspect that sheds light on the efficiency of regression estimation. While drawing comparisons to the Frisch-Waugh-Lovell theorem, we argue that understanding the possibility of estimating regression in blocks--with minimal information required to be saved between iterations--is of significant interest. The Frisch-Waugh-Lovell theorem, known for its column-wise approach, provides a technique to reduce the total dimension of $K$ in calculations. In contrast, our paper introduces a `transposed' version, presenting a row-wise result that allows for the reduction of the total dimension of $N$ in calculations. Generally, and especially with the growing relevance of high-frequency and administrative datasets in economic research, $N$ is substantially greater than $K$, meaning that reductions in $N$ can lead to far greater computational savings than reductions in $K$. Notably, both the Frisch-Waugh-Lovell theorem and our paper  are motivated by computational considerations, highlighting the shared emphasis on addressing computational challenges in econometrics \cite{Mackinnon2022}.  However, beyond the computational elements of the paper, both of these results are theoretically elegant, and provide an understanding of the internal workings of regression models: one of the most commonly used tools of researchers in all empirical fields of economics.

This paper joins other studies from a range of settings which provide basic understanding of regression-based models (see for example \citeN{Slocynski2022,Abadie2003,StefanskiBoos2002,Gelbach2016,Angrist1998,Abadieetal2020,Solonetal2015}).  It also provides results which are potentially highly useful in cases where very large databases are used in econometric analysis.  Massive databases are increasingly common in econometric analyses \cite{Varian2014}, but supercomputers are not always available to process big chunks of data.  Indeed, \citeN{Mackinnon2022} calls for consideration of computational issues with large datasets, noting that ``[i]n recent years (\ldots) many interesting datasets seem to be becoming larger more quickly than computers are becoming faster.'' While true, the results in this paper suggest that many of the processes of interest in econometrics can be implemented in a partition-wise fashion, implying that memory costs can be avoided.  While an alternative solution to these issues is to simply gain access to super-computers or large server clusters, this solution may be infeasible for individuals with small research budgets or students, who nevertheless wish to use large datasets.  These results can thus also be viewed as democratising access to econometric tools.   Finally, we note that these results can offer substantial speed-ups for clustered bootstrapping, joining a literature which considers the computational efficiency of bootstrap procedures, and clustering in particular (see eg \shortciteN{Cameronetal2008,Roodmanetal2019,MacKinnon2023}), as well as for the consideration of tuning parameters in regularised regression models.

This paper is structured as follows.  In Section \ref{scn:CLS} we define the cumulative least squares procedure, showing its equivalence to standard estimation.  We begin by showing how this estimator works in cases where estimation proceeds by OLS assuming homoscedasticity, and then document how it holds in a broad range of other estimation and inference procedures.  Section \ref{scn:optimal} discusses the nature of the cumulative procedure, and considerations of optimal block sizes for estimation.  In Section \ref{scn:illustrations} we provide a number of illustrations of the performance of these methods compared to commonly-used commercial alternatives.  This includes controlled tests where sample sizes and covariate numbers are varied and computational efficiency is compared, as well as an applied example based on a sample of census data and demographic surveys and models with a large number of fixed effects.  In Section \ref{scn:conclusion} we provide some additional discussion and conclusions.

\section{Cumulative Least Squares}
\label{scn:CLS}
\subsection{Cumulative Ordinary Least Squares}
\label{sscn:COLS}
Suppose we wish to run a regression of a dependent variable $y$ on a set of $K$ covariates $x_1,x_2,...,x_K$, using a series of $i=1,\ldots,N$ observations.  Thus, data can be viewed as a matrix or database of size $N\times K$ independent variables which we will denote $X$, as well as an $N\times1$ vector $y$ for the dependent variable.  Throughout this paper we will adopt the notation that matrices are written as upper case italics, vectors are written as lowercase italics, and scalars are defined as required. Suppose also that computing the regression with all the data in memory is either infeasible or undesired due to memory constraints. The data can be partitioned row-wise in $J$ arbitrarily defined portions, where each portion, or block, is denoted $j$, and consists of $N_j=N/J$ observations.\footnote{To fix ideas, we will consider that $N_j$ is common across all blocks.  In Section \ref{scn:optimal} we will discuss optimal choices of $N_j$, not requiring that this number be equivalent across blocks.}  The blocks are mutually exclusive and cover all observations such that $\sum_{j=1}^JN_j=N$.  We use the notation $X^j$ to denote block $j$ of size $N_j$ of the independent variables, and similarly $y^j$ is used to denote block $j$ of the dependent variable.

Consider the OLS estimator of the parameter $\widehat\beta$.  The standard OLS estimator can be written as follows:
\footnotesize
\begin{eqnarray}
\label{eqn:CLS1} 
\widehat\beta_{OLS}=(\XpX)^{-1}\Xpy &\equiv&  \left( \left(
  \begin{array}{cccc}
    X^{1\prime} & X^{2\prime} & \cdots & X^{J\prime} \\
  \end{array}\right)\left(
  \begin{array}{c}
    X^1 \\
    X^2 \\
    \vdots \\
    X^J
  \end{array}
  \right)
  \right)^{-1}  \left(
  \begin{array}{cccc}
    X^{1\prime} & X^{2\prime} & \cdots & X^{J\prime}\\
  \end{array}\right) \left(
  \begin{array}{c}
    y^1 \\
    y^2 \\
    \vdots \\
    y^J
  \end{array}
  \right) \\ \label{eqn:CLS2} 
  &=& \Bigl(\XpXj{1} + \XpXj{2} + \cdots + \XpXj{J}\Bigr)^{-1}\Bigl(\Xpyj{1} + \Xpyj{2} + \cdots + \Xpyj{J}\Bigr)
\end{eqnarray}
\normalsize
where in (\ref{eqn:CLS1}) the $K\times N$ matrix $X^\prime$ is re-expressed (identically) as a series of horizontally concatenated $K\times N_j$ matrices, and the $N\times K$ matrix $X$ is similarly re-expressed as a series of vertically concatenated $N_j\times K$ matrices.  The $N \times 1$ vector $y$ is also re-written as a series of vertically stacked sub-vectors  of dimension $N_j \times 1$. Based on the properties of matrix multiplication, it can easily be seen that elements from each sub-matrix or vector will be interacted only with themselves, and no products are required across blocks.  The re-expressed version of $\widehat\beta_{OLS}$ in (\ref{eqn:CLS2}) makes clear that $\XpX$ can thus be re-written as the summation over the series of $J$ matrices $\XpXj{j}$ which are each of dimension $K\times K$, and a similar procedure can be followed for $\Xpy$.  

This suggests that a cumulative procedure can be followed, as laid out formally in Algorithm \ref{alg:CLS} below.  Specifically, for ease of notation denote $\XpXj{j}\equiv\Sigma_j$ and $\Xpyj{j}\equiv\Upsilon_j$. Define as $\Sigma_{1\sim j}$ the summation $\Sigma_1+\ldots\Sigma_j$, and $\Upsilon_{1\sim j}=\Upsilon_1+\ldots+\Upsilon_j$. 
Then, to estimate (\ref{eqn:CLS2}), initially a single block of data can considered, and the quantities $\Sigma_1$ and $\Upsilon_1$ calculated.  In the following step, a new block of data can be consulted, the quantities $\Sigma_2$ and $\Upsilon_2$ calculated, and the preceding quantities summed to provide $\Sigma_{1\sim 2}$ and $\Upsilon_{1\sim 2}$. In following steps, accumulated quantities $\Sigma_{1\sim j-1}$ and $\Upsilon_{1\sim j-1}$ are received at the beginning of each stage, the $\Sigma_{j}$ and $\Upsilon_{j}$ are calculated, and the step ends with $\Sigma_{1\sim j}$ and $\Upsilon_{1\sim j}$. A key element of this procedure is that in each stage, only a single block of data of size $N_j\times (K+1)$ needs to be read into memory, with the results stored in a single accumulated matrix and vector $\Sigma_{1\sim j}$ and $\Upsilon_{1\sim j}$.  As $\Sigma$ and $\Upsilon$ are of dimensions $K\times K$ and $K\times 1$ respectively, this makes clear that we simply need to keep track of small matrices in an ongoing fashion, and never house more than $N_j$ observations in memory at a single time, where $N_j$ can be an arbitrarily small value (in the limit, this could even be 1).\footnote{This particular limit case where $N_j=1$ and estimation occurs via OLS to generate the matrix $\Sigma_{1\sim J}$ is mentioned in \citeN{Brownetal1953}.} The OLS estimate $\widehat\beta_{OLS}$ is \emph{only} calculated by matrix inversion (or alternative procedures such as Gauss-Jordan elimination) once the full matrices $\Sigma_{1\sim J}\equiv\XpX$ and $\Upsilon_{1\sim J}\equiv\Xpy$ are calculated, implying that potentially costly matrix inversions are not required at every step.

The above process allows for point estimates to be recovered from independent partitions of the database.  What's more, inference on regression parameters can be conducted in a similar partition-wise manner.  Assuming homoscedasticity (alternative inference procedures are considered in Section \ref{sscn:inference}), the well-known formula for the variance of OLS regression parameters is $\widehat{V}(\widehat\beta_{OLS})=\widehat{\sigma}^2_u(X^\prime X)^{-1}$.
The quantity $(X^\prime X)$ is already accumulated as laid out above.  The second element of the variance is  $\widehat{\sigma}^2_u\equiv \widehat{u}^\prime\widehat{u}/(N-K)$, where the regression residuals $\widehat{u}=y-X\widehat\beta_{OLS}=(I-X(\XpX)^{-1}X^\prime)y=M_Xy$, with $M_X$ being the annihilator matrix, an idempotent matrix. Hence:
\begin{equation}
\label{eqn:UU}
\widehat{u}'\widehat{u} = y^\prime M_X y = y^\prime y - y^\prime X (\XpX)^{-1} \Xpy,
\end{equation}
which consists of three separate elements: $\XpX$, $\Xpy$ and its transpose, and $\Ypy$.  The first two of these elements are already calculated iteratively in the estimation of point estimates as $\Sigma_{1\sim J}$ and $\Upsilon_{1\sim J}$.  The only additional element required to calculate $\widehat\sigma^2_u$ is thus $\Ypy$, which can similarly be calculated in a cumulative manner in the same fashion as $\XpX$ or $\Xpy$ in (\ref{eqn:CLS2}): $y^\prime y=(\Ypyj{1}+\Ypyj{2}+\cdots+\Ypyj{J})$.  As above, we will refer to $\Ypyj{j}\equiv\Psi_j$, and $\Psi_{1\sim j}=\Psi_1+\ldots +\Psi_j$.  Hence, calculating $\Ypy$ occurs iteratively, where at each step the accumulated $\Psi_{1\sim j-1}$ is the starting point, an additional block of $y$ of dimension $N_j\times 1$ is loaded, and the step ends with $\Psi_{1\sim j}=\Psi_{1\sim j-1}+\Psi_{j}$.\footnote{In Appendix \ref{scn:ULS} we note that this result can be documented in an alternative way, where rather than accumulating matrices $\Sigma, \Upsilon$, and $\Psi$ at each step, the estimate $\widehat\beta_{1\sim j}$ is directly updated.  This result is based on the matrix inverse lemma \cite{Woodbury50}.  However, given that this is less efficient than the cumulative procedure described here, we document this only as a curiosity.}

Formally, the entire estimation process to arrive to exact OLS point estimates and standard errors is laid out in Algorithm \ref{alg:CLS}.  Note that given the information calculated in Algorithm \ref{alg:CLS}, other standard regression statistics can be generated following estimation, including $t$-tests for each regression parameter against arbitrary null hypotheses, global $F$-tests of regressions, and $R^2$ or adjusted $R^2$ measures. For example, in order to compute the $R^2$, we can use the residual sum of squares (RSS), $\widehat{u}^\prime \widehat{u}$ calculated above, and additionally require the total sum of squares (TSS), given that $R^2=1-\frac{RSS}{TSS}$.  The TSS is simply:
\[
TSS = \sum_{i=1}^N \left(y_i - \overline{y}\right)^2=\sum_{i=1}^N y_i^2 - N\left(\frac{1}{N}\sum_{i=1}^N y_i \right)^2 = \sum_{i=1}^N y_i^2 - \frac{1}{N}\left(\sum_{i=1}^N y_i \right)^2,
\]
and both $y_i^2$ and $y_i$ can be summed iteratively, with the only addition to the statistics already laid out above being the cumulative sum of $y$ squared, and grand mean $\overline{y}$. 

\begin{algorithm}
\caption{Cumulative ordinary least squares}
\label{alg:CLS}
Inputs: Database consisting of ($y$,$X$), block size $b$. \\
Result: Point estimate $\widehat\beta_{OLS}$ and variance-covariance matrix $\widehat{V}(\widehat\beta_{OLS})$.  \\ \ \\
1. Set $i=1$ and $j=b$. Load into memory partition of data covering $y,X$ in observations $i$ to $j$. \;
2. Calculate $\Sigma_1$, $\Upsilon_1$, and $\Psi_i$ \;
3. If observations $i$ to $j$ contain end of file, set $e=1$, otherwise, set $e=0$ \;
\While{$e\neq1$}
{
4. Replace $i=i+b$ and $j=j+b$. Load into memory partition of data covering $y,X$ in observations $i$ to $j$. \;
5. Calculate $\Sigma_j$, $\Upsilon_j$, and $\Psi_j$. \;
6. Calculate $\Sigma_{1\sim j}=\Sigma_{1\sim j-1}+\Sigma_j$,
$\Upsilon_{1\sim j}=\Upsilon_{1\sim j-1}+\Upsilon_j$, 
and
$\Psi_{1\sim j}=\Psi_{1\sim j-1}+\Psi_j$ \;
7. If observations $i$ to $j$ contain end of file, set $e=1$. \;
}\textbf{end}\\ \ \\
8. Calculate \\
\[\widehat{\beta}_{OLS}=\left(\Sigma_{1\sim J}\right)^{-1}\Upsilon_{1\sim J} \qquad \text{ and } \qquad \widehat{V}\left(\widehat\beta_{OLS}\right)=\widehat\sigma^2(X'X)^{-1}, \] where $\widehat\sigma^2=\left[\Psi_{1\sim J}-\Upsilon^\prime_{1\sim J}(\Sigma_{1\sim J})^{-1}\Upsilon_{1\sim J}\right]/(N-K)$.\\
\end{algorithm}

\subsection{Alternative Estimators}
\label{sscn:alternative}
While the previous implementation allows for the generation of exact equivalents to OLS estimates and their standard errors (and derived statistics), this cumulative procedure can be applied far more widely.  Indeed, the procedure can be used for \emph{any} estimator which can be expressed as a sum of squares-based procedure, where the relevant database level processing requires sums over observation-level products.  We document how the cumulative process works in a range of estimators below.  We then document that similar logic can be used to arrive to estimates which are based upon other techniques such as maximum likelihood.

\paragraph{Weighted Least Squares}
A simple extension to the procedure noted in section \ref{sscn:COLS} is weighted least squares, where some diagonal weight matrix $W$ is incorporated, such that the estimator is defined as:
\begin{equation}
\small
\label{eq:WLS}
\widehat\beta_{WLS}=(X^\prime W X)^{-1}X^\prime W y =\left(\XpWXj{1} + \cdots + \XpWXj{J}\right)^{-1}\left(\XpWyj{1} + \cdots + \XpWyj{J}\right).
\end{equation}
Here, it follows that an identical updating procedure can be implemented to that laid out in the case of OLS, however additionally, a variable $w$ contains the weight associated with each observation.    In this case, the cumulative estimation procedure consists of holding in memory a single block of data $(y_j, w_j, X_j)$ and generating matrix $W_j$, an $N_j\times N_j$ matrix with elements $w_j$ on the principal diagonal.  In the limit, if $N_j=1$, the matrix $W_j$ consists simply of the scalar $w_j$.  Then, elements $\XpWXj{j}$ and $\XpWyj{j}$ are calculated, and summed cumulatively, before in a final step the WLS estimator is calculated by matrix inversion or similar.






\paragraph{Instrumental Variables and Two-Stage Least Squares Estimators}
Both instrumental variables (IV) and Two-Stage Least Squares (2SLS) estimators can be similarly estimated in cumulative form.  To see this, note that the IV estimator in a linear model is
$\widehat\beta_{IV}=(\XpZ)^{-1}\Zpy$ and the 2SLS estimator in a linear model is: $\widehat{\beta}_{2 SLS}=\left(\XpZ\left(Z^{\prime} Z\right)^{-1} Z^{\prime} X\right)^{-1}\left(\XpZ \left(Z^{\prime} Z\right)^{-1} \Zpy\right)$, 
where $Z$ refers to an $N\times L$ dimensional vector of exogenous variables, with $L\geq K$, and in the case of IV, $L=K$.  Thus, both $\widehat\beta_{IV}$ and $\widehat\beta_{2SLS}$ can be generated cumulatively following a similar procedure to \eqref{eqn:CLS1}, however in the case of $\widehat\beta_{IV}$ $\XpX$ is substituted for $\XpZ$, and $\Xpy$ is substituted for $\Zpy$.  In the case of 2SLS, an additional quantity $Z^\prime Z$ must be calculated, though identically to $\XpX$, this simply requires cross-products on all variables $Z$ within each observation $i$, and as in \eqref{eqn:CLS2}, $\ZpZ\equiv(\ZpZj{1}+\ZpZj{2}+\ldots+\ZpZj{J})$.  Once again, estimation can proceed in this case in a cumulative fashion, where in each block the quantities $\XpZj{j}$, $\ZpZj{j}$ and $\Zpyj{j}$ are calculated, summed cumulatively, and ultimately, the quantity $\widehat\beta_{2SLS}$ is calculated by matrix inversion and multiplication, or other standard procedures such as QR decomposition or single value decomposition.  


\paragraph{Ridge, LASSO and Elastic Net}
Frequently in cases where big data is used in economic models, practitioners wish to perform some sort of regularisation.  Fortunately, these cumulative procedures cross-over seamlessly to regularised models such as Ridge, LASSO and Elastic Net.  Additionally, in each case, the process of accumulation is such that work with the full dataset of dimension $N\times K$ can be viewed as a first data processing step, and the selection of tuning parameters can be conducted as a second step, without ever returning to full data.  

To see this, we first document the case of the Ridge regression, which given its use of the $\ell^2$ norm for shrinkage is particularly simple expositionally.  In the case of the Ridge regression, parameters are estimated as follows:
\[
\widehat\beta_{Ridge} = \argmin_\beta\left\{\left(\sum_{i=1}^N(y_i-X^\prime_i\beta)\right)+\lambda\sum_{j=1}^K\beta_j^2\right\},
\]
where $\lambda$ is a scalar tuning parameter determining the degree of shrinkage. This can equivalently be written as:
\begin{equation}
\label{eqn:ridge}
\widehat\beta_{Ridge} = (\XpX+\lambda I)^{-1}\Xpy 
\end{equation}
where $I$ is an identity matrix of size $K$.  Note that following the notation above, solving for $\widehat\beta_{Ridge}$ requires the quantity $\Sigma_{1\sim J}$, which we have documented can be calculated in a cumulative fashion, $\Upsilon_{1\sim J}$, which we have also documented can be calculated cumulatively, and an additional factor $\lambda I$, which is independent of the number observations.  Hence, estimation in the case of Ridge is identical to that documented in OLS in section \ref{sscn:COLS}, with the only difference being that after accumulating $\Sigma_{1\sim J}$ and $\Upsilon_{1\sim J}$, but prior to solving for $\widehat\beta_{Ridge}$, an additional $K\times K$ matrix is added to $\Sigma_{1\sim J}$.  This is laid out formally in Algorithm \ref{alg:CLSridge}.

\begin{algorithm}
\caption{Cumulative least squares for Ridge regression}
\label{alg:CLSridge}
Inputs: Database consisting of ($y$,$X$), block size $b$. \\
Result: Point estimate $\widehat\beta_{Ridge}$.  \\ \ \\
1. Set $i=1$ and $j=b$. Load into memory partition of data covering $y,X$ in observations $i$to $j$. \;
2. Calculate $\Sigma_1$ and $\Upsilon_1$. \;
3. If observations $i$ to $j$ contain end of file, set $e=1$, otherwise, set $e=0$ \;
\While{$e\neq1$}
{
4. Replace $i=i+b$ and $j=j+b$. Load into memory partition of data covering $y,X$ in observations $i$ to $j$. \;
5. Calculate $\Sigma_j$ and $\Upsilon_j$. \;
6. Calculate $\Sigma_{1\sim j}=\Sigma_{1\sim j-1}+\Sigma_j$ and
$\Upsilon_{1\sim j}=\Upsilon_{1\sim j-1}+\Upsilon_j$. \;
7. If observations $i$ to $j$ contain end of file, set $e=1$. \;
}\textbf{end}\\ \ \\
8. Select tuning parameter $\lambda$.  Then
calculate \\
\[\widehat{\beta}_{Ridge}=(\Sigma_{1\sim J} + \lambda I)^{-1}\Upsilon_{1\sim J}. \]\\
\end{algorithm}

Similar procedures can be conducted in the case LASSO and Elastic net, where first data can be accumulated to form $\Sigma_{1\sim J}$ and $\Upsilon_{1\sim J}$ and then, conditional on having processed the data of dimension $N\times K$ to a level of $K\times K$ (or $K\times 1$), and selecting a tuning parameter\footnote{We note below that our procedures can similarly be used very efficiently for $k$-fold cross validation; see Section \ref{sscn:groups}.},  estimates are calculated without ever returning to data at a level of $N\times K$ (or even $N_J\times K$).  To see this, note that the Lasso and Elastic net equivalents of (\ref{eqn:ridge}) are:
\begin{eqnarray}
\widehat\beta_{Lasso} &=& \argmin_\beta\left\{\left(\sum_{i=1}^N(y_i-X^\prime_i\beta)\right)+\lambda||\beta_j||_1\right\}, \nonumber \\
\widehat\beta_{Elastic Net} &=& \argmin_\beta\left\{\left(\sum_{i=1}^N(y_i-X^\prime_i\beta)\right)+\lambda_1||\beta_j||_1+\frac{\lambda_2}{2}||\beta_j||^2_2\right\}, \nonumber 
\end{eqnarray}
where $||\cdot||_p$ refers to the $\ell^p$ norm, and in the case of the Elastic net, $\lambda_1$ and $\lambda_2$ refer to the strength of the Lasso and Ridge penalties respectively.

Although the lack of the exclusive $\ell^2$ norm in Lasso and Ridge does not admit a simple least-squares solution as in (\ref{eqn:ridge}), they nevertheless can both be simply resolved using cumulative procedures and a single (accumulatory) pass through $N$ dimensional data.  Specifically, this can be implemented via coordinate descent, a standard way of computing parameters in Lasso and Elastic Net \cite{Fu1998}.  To see this note that for a specific parameter $\beta_j$, the coordinate descent algorithm for estimation can be written for Lasso as:
\[
\widehat\beta^{\text{new}}_j = \text{sign}\left(\widehat\beta_j^{\text{old}}\right)\text{max}\left(|z_j|-\frac{\lambda}{N},0\right) 
\]
where $\widehat\beta_j^{\text{old}}$ is the value of $\widehat\beta_j$ at the previous iteration, $z_j$ is the $j$\textsuperscript{th} element of the vector $z=X^\prime(y-X\widehat\beta^{\text{old}})$ and sign($\cdot$) returns the sign of the argument.  Noting that $z$ can be re-expressed as $X^\prime y - \XpX\widehat\beta^{\text{old}}$ makes clear that the vector of parameters $\beta_j$ can be estimated by first using the cumulative procedure laid out previously, and then working with $\Xpy$ and $\XpX$ in coordinate descent, without ever returning to the original data.  A similar procedure can be used for Elastic Net given that in this case successive iterations of coordinate descent can be calculated as:
\[
\widehat\beta^{\text{new}}_j = \frac{z_j}{1+\lambda_2}\left(\text{sign}\left(\widehat\beta_j^{\text{old}}\right)\text{max}\left(|z_j|-\frac{\lambda_1}{1+\lambda_2},0\right)\right),
\]
where $z_j$ is the $j^{th}$ element of  $z=X^\prime(y-X\widehat\beta^{\text{old}})+\widehat\beta^{\text{old}}\cdot\lambda_2$, and $\lambda_1$ and $\lambda_2$ are Lasso and Ridge regularisation parameters.  Again, given that $z$ can be expressed as $X^\prime y - \XpX\widehat\beta^{\text{old}}+\widehat\beta^{\text{old}}\cdot\lambda_2$, estimation can proceed by, firstly, accumulating 
$\Sigma_{1\sim J}$ and $\Upsilon_{1\sim J}$ in a block-by-block or line-by-line fashion, and then implementing coordinate descent with, at most, matrices of dimension $K\times K$.

\paragraph{Binary Choice Models via Iteratively Reweighted Least Squares}
The previously defined estimators can be implemented in a single cumulative step, potentially offering substantial speed-ups compared to traditional estimators in cases where both cumulative and standard estimators are feasible, but where limits are close to met when all observations are housed in working memory (further discussion on relative performance of cumulative and naive procedures are provided in the following sections).  In the case of Binary Choice Models such as probit and logit models, cumulative procedures can similarly be implemented which exactly replicate non-cumulative procedures while at the same time never housing more than a small number of observations in memory.  However in these cases it is not possible to implement these estimators as single shot processes, but rather multiple passes through the $N$ rows of data must be conducted.  Thus, while these procedures provide feasible implementations of estimators when the entire dataset cannot be held in a computer's working memory, they are unlikely to be as fast as standard procedures when memory is not a limiting factor.

Nevertheless, to see that cumulative procedures can also be implemented in non-linear models, one alternative is to use Iteratively Reweighted Least Squares (IRLS).  IRLS allows for the estimation of the parameters in non-linear models in a step-wise fashion, where at each step the updated parameter estimates are based on a weighted least squares problem \cite{NelderWedderburn1972,Green1984}.  Based on this, cumulative procedures can be used to conduct least squares estimators in each iteration. Specifically, the IRLS procedure for binary outcome models consists of iteratively solving the following equation until $\widehat\beta^{\text{old}}$ and $\widehat\beta^{\text{new}}$ converge:
\begin{eqnarray}
\nonumber
\widehat\beta^{\text{new}}&=&(X^\prime WX)^{-1}X^\prime WZ \quad \text{where} \quad Z=X\widehat\beta^{\text{old}}+W^{-1}(y-p) \\
&=& \widehat\beta^{\text{old}} + (X^\prime WX)^{-1}X^\prime(y-p)
\label{eq:IRLS}
\end{eqnarray}
Here $y$ is an $N\times 1$ vector of outcome variables, and $p$ is predicted value for each unit $p_i(x_i,\widehat\beta^{\text{old}})$ based on the $1\times K$ vector of individual-level realisations $x_i$, such that $y-p$ represents prediction residuals.  $W$ is an $N\times N$ diagonal weight matrix with diagonal elements consisting of $p_i(x_i,\widehat\beta^{\text{old}})(1-p_i(x_i,\widehat\beta^{\text{old}}))$.  In the case of probit models, for example $p(x_i,\beta)\equiv\phi(x_i\beta)$, while in the case of logit models, $p(x_i,\beta)=\ln(x_i\beta/(1-x_i\beta))$.
The quantity in \eqref{eq:IRLS} consists of some starting value  $\widehat\beta^{\text{old}}$ which is taken as an input (in the first iteration, $\widehat\beta^{\text{old}}=0$), and a second component which can be calculated cumulatively in a block-wise fashion following \eqref{eq:WLS}. Thus one can estimate non-linear models where in each step a cumulative procedure is performed, and a solution is reached when the second term in \eqref{eq:IRLS} converges to 0.


\paragraph{Maximum Likelihood and other M-Estimators}
The use of cumulative procedures like those described above can similarly be employed to with other classes of M-estimators where estimation is based on iterative optimisation procedures provided that observations are assumed to be independently sampled.\footnote{In cases where sampling is not assumed to be independent, generalisations of this procedure could be followed, but likelihood functions, and hence blocks in the data in cumulative procedures, would need to permit this dependence.  We discuss one such case where sampling is not assumed to be independent in Section \ref{sscn:inference} below.}  To see this, consider maximum likelihood estimation implemented using the Newton-Raphson method.  Estimation occurs iteratively, where at each stage the Hessian and Score matrices are evaluated based on the current iteration of $\beta$.  Specifically, estimation occurs as follows:
\begin{equation}
\label{eq:ML}
\widehat\beta^{\text{new}}=\widehat\beta^{\text{old}}-\left[\frac{\partial^2 \ell(\beta)}{\partial\beta\partial\beta^\prime}\right]_{\beta=\widehat\beta^{\text{old}}}^{-1}\left[\frac{\partial\ell(\beta)}{\partial\beta}\right]_{\beta=\widehat\beta^{\text{old}}},
\end{equation}
with the ML solution occurring when this equation converges.

When observations are independent, the Hessian and score matrices in ML are written as summations over observations $i$. For example, in the case of the logit regression:
\[
\frac{\partial\ell(\beta)}{\partial\beta}=\sum_{i=1}^N\bigg[y_iF(-x_i\beta)-(1-y_i)F(x_i\beta)\bigg]x^\prime_i \qquad \frac{\partial^2 \ell(\beta)}{\partial\beta\partial\beta^\prime}=-\sum_{i=1}^Nf(x_i\beta)x_i^\prime x_i,
\]
where $F(\cdot)$ and $f(\cdot)$ are the logit cdf and pdf respectively.\footnote{Similar examples can be easily provided for other common models estimated via ML.  In the case of the probit regression, these functions are written as summations over $i$ of the following form:
\begin{eqnarray}
\frac{\partial\ell(\beta)}{\partial\beta}&=&\sum_{i=1}^N\bigg[y_i\frac{\phi(x_i\beta)}{\Phi(x_i\beta)}-(1-y_i)\frac{\phi(x_i\beta)}{1-\Phi(x_i\beta)}\bigg]x^\prime_i  \nonumber \\ \frac{\partial^2 \ell(\beta)}{\partial\beta\partial\beta^\prime}&=&-\sum_{i=1}^N\phi(x_i\beta)\bigg[y_i\frac{\phi(x_i\beta)+x_i\beta\Phi(x_i\beta)}{\Phi(x_i\beta)^2}-(1-y_i)\frac{\phi(x_i\beta)-x_i\beta(1-\Phi(x_i\beta))}{[1-\Phi(x_i\beta)]^2}\bigg]x^\prime_ix_i.\nonumber
\end{eqnarray}
where $\phi(\cdot)$ and $\Phi(\cdot)$ are the normal pdf and cdf respectively.}
This suggests a cumulative procedure can be employed where a block of arbitrary size $N_j$ can be read in to memory and the Hessian and Score matrix can be calculated for this block $j$ based on the values $\beta=\widehat\beta^{old}$.  The summation for each matrix can be stored, and then a subsequent block of size $N_j$ can be read in, the Hessian and Score matrices can be calculated, and added to the previous values.  This process can be updated cumulatively until the end of the data is reached.  Finally, a new value for $\widehat\beta$ can be calculated as in \eqref{eq:ML}, either providing the ML estimate if convergence has occurred, otherwise data will be read again, and another iteration of \eqref{eq:ML} calculated.   In this case, as noted previously with IRLS, this procedure is feasible when large databases cannot be read into memory in their entirety, but is unlikely to be as fast as a standard ML procedures if the entire data \emph{can} be stored in memory.

\subsection{Grouped Estimation Procedures, Fixed Effect Estimators, Heterogeneity, and Cross-Validation}
\label{sscn:groups}
In the previous section, results were shown based on arbitrary divisions of the data into mutually exclusive blocks.  All of the previous results hold if rather than groups of data being based on positions, groups of data are based on some particular indicator.  Consider a variable $G$ capturing membership in some particular group, with group levels $g\in\mathcal{G}$.  Using notation $X_{g}$ and $y_{g}$ to indicate realisations of $X$ and $y$ respectively for observations where $G=g$, it is well known that the OLS estimate $\widehat\beta_{OLS}$ can be generated over groups as:
\begin{equation}
\label{eqn:grouped}
\widehat\beta_{OLS}=\left(\sum_{g\in\mathcal{G}}X_{g}^{\prime}X_{g}\right)^{-1}\left(\sum_{g\in\mathcal{G}}X_{g}^{\prime}y_{g}\right)
\end{equation}
What's more, as was the case previously, quantities $X_{g}^{\prime}X_{g}$ and $X_{g}^{\prime}y_{g}$ can be built up cumulatively from arbitrarily small portions of data.

In practice, this is simply a group-level generalisation of the procedure described in Algorithm \ref{alg:CLS}.  As in Section \ref{sscn:COLS}, consider data broken down into $J$ row-wise partitions, with each block denoted $j$ and consisting of $N_j$ observations.  For a particular group $g\in\mathcal{G}$ define $X_{g}^{j\prime}X^j_{g}\equiv\Sigma^{g}_j$, and similarly, $X_{g}^{j\prime}y^j_{g}\equiv\Upsilon^{g}_j$.  If no observations for group $g$ are present in block $j$, $\Sigma^{g}_j$ is simply defined to be a null matrix $O_{K,K}$, and  $\Upsilon^{g}_j$ a null vector $O_{K,1}$. As previously, $\Sigma^{g}_{1\sim j}$ refers to the summation $\Sigma^{g}_1+\ldots\Sigma^{g}_j$, and $\Upsilon^{g}_{1\sim j}=\Upsilon^{g}_1+\ldots+\Upsilon^{g}_j$.  A group-level generalisation of Algorithm \ref{alg:CLS} is described in Algorithm \ref{alg:GCLS} below.

\begin{algorithm}[ht!]
\caption{Grouped cumulative ordinary least squares}
\label{alg:GCLS}
Inputs: Database consisting of ($y$,$X$,$G$), block size $b$. \\
Result: Point estimate $\widehat\beta_{OLS}$ and variance-covariance matrix $\widehat{V}(\widehat\beta_{OLS})$. Aggregates $\Sigma^{g}_{1\sim J}$ and  $\Upsilon^{g}_{1\sim J}$ $\forall g\in\mathcal{G}$.  \\ \ \\
1. Set $i=1$ and $j=b$. Load into memory partition of data covering $y,X,G$ in observations $i$ to $j$. \;
\For{$g\in\mathcal{G}^{1}$}
{
2. Calculate $\Sigma^{g}_1$, $\Upsilon^{g}_1$, and $\Psi^{g}_i$ \;
}
3. If observations $i$ to $j$ contain end of file, set $e=1$, otherwise, set $e=0$ \;
\While{$e\neq1$}
{
4. Replace $i=i+b$ and $j=j+b$. Load into memory partition of data covering $y,X,G$ in observations $i$ to $j$. \;
\For{$g\in\mathcal{G}^{j}$}
{
5. Calculate $\Sigma^{g}_j$, $\Upsilon^{g}_j$, and $\Psi^{g}_j$. \;
6. If $\exists\ \Sigma^{g}_{1\sim j-1}$ calculate $\Sigma^{g}_{1\sim j}=\Sigma^{g}_{1\sim j-1}+\Sigma^{g}_j$,
$\Upsilon^{g}_{1\sim j}=\Upsilon^{g}_{1\sim j-1}+\Upsilon^{g}_j$, 
and
$\Psi^{g}_{1\sim j}=\Psi^{g}_{1\sim j-1}+\Psi^{g}_j$, otherwise, initialise $\Sigma^{g}_{1\sim j}=\Sigma^{g}_j$,
$\Upsilon^{g}_{1\sim j}=\Upsilon^{g}_j$, 
and
$\Psi^{g}_{1\sim j}=\Psi^{g}_j$ \;
}
7. If observations $i$ to $j$ contain end of file, set $e=1$. \;
}\textbf{end}\\ \ \\
8. Calculate $\Sigma_{1\sim J}=\sum_{g\in\mathcal{G}}\Sigma^{g}_{1\sim J}$, $\Upsilon_{1\sim J}=\sum_{g\in\mathcal{G}}\Upsilon^{g}_{1\sim J}$, and $\Psi_{1\sim J}=\sum_{g\in\mathcal{G}}\Psi^{g}_{1\sim J}$.  Then: \\
\[\widehat{\beta}_{OLS}=(\Sigma_{1\sim J})^{-1}\Upsilon_{1\sim J} \qquad \text{ and } \qquad \widehat{V}(\widehat\beta_{OLS})=\widehat\sigma^2\Sigma^{-1}_{1\sim J}, \] where $\widehat\sigma^2=\left[\Psi_{1\sim J}-\Upsilon^\prime_{1\sim J}(\Sigma_{1\sim J})^{-1}\Upsilon_{1\sim J}\right]/(N-K)$.\\
\end{algorithm}

\paragraph{Heterogeneity} An immediate implication of this group-level cumulative procedure is that instead of generating a single $K\times K$ matrix $\Sigma_{1\sim J}$ and $K\times 1$ vector $\Upsilon_{1\sim J}$, $N_G$ versions of these matrices will be generated, where $N_G$ refers to the distinct number of groups.  Estimation of overall OLS parameters can then occur following \eqref{eqn:grouped}, or any other grouped level estimator can similarly be generated.  However, given that group level statistics $\Sigma^{g}_{1\sim J}$ and $\Upsilon^{g}_{1\sim J}$ are also generated, identical models for any sub-samples can then be generated nearly instantaneously, without ever returning to individual level data.  This includes estimates for each specific group $g$, but also for aggregated groups, such as groups of states or groups of countries. In Section \ref{sscn:inference} we will return to show that this also offers substantial benefits for inference in cases of (blocked) bootstrap procedures.

\paragraph{Fixed Effect Estimators}
We can similarly use these group-level procedures to generate fixed-effect estimators, again without ever returning to individual level data.  To see how fixed effect estimators can also be estimated in a cumulative fashion, we will now double-index as $Y_{gt}$ an observation $t$ within group $g$ (this can be thought of, for example, as a case where observations are repeated within $g$ across time periods denoted $t$).  We are interested in estimating the parameter vector on some independent variables $X_{gt}$, while controlling for time-invariant group fixed effects $\mu_g$.  The fixed effect estimator can be generated from an OLS regression on within transformed data.  Specifically, this consists of estimating:
\begin{eqnarray}
y_{gt}-\bar{y}_g&=&(X_{gt}-\bar{X}_g)\beta_{FE}+(\mu_g-\bar{\mu}_g)+u_{gt}-\bar{u}_g \nonumber\\
\dot{y}_{gt}&=&\dot{X}_{gt}\beta_{FE}+\dot{u}_{gt}\nonumber
\end{eqnarray}
where $\dot{y}_{gt}$ denotes the within transformation of $y$, $\bar{y}_g$ refers to group-level means, and similarly for other variables. The term $u_{gt}$ is a time-varying stochastic error.  The fixed effect estimator is then written as below:
\begin{eqnarray}
\label{eqn:FE}
\widehat\beta_{FE}&=&\left(\sum_{g\in\mathcal{G}}\sum_{t=1}^T\dot{X}_{gt}^\prime \dot{X}_{gt}\right)^{-1}\left(\sum_{g\in\mathcal{G}}\sum_{t=1}^T\dot{X}_{gt}\dot{y}_{gt}\right) \nonumber \\
&=&\sum_{g\in\mathcal{G}}\sum_{t=1}^T\left(X_{gt}^\prime X_{gt}-\bar{X}_g^\prime\bar{X}_g\right)^{-1}\sum_{g\in\mathcal{G}}\sum_{t=1}^T(X_{gt}^\prime y_{gt}-\bar{X}_g^\prime\bar{y}_g).
\end{eqnarray}
The key insight in \eqref{eqn:FE} is that $\sum_{g\in\mathcal{G}}\sum_{t=1}^T(X_{gt}-\bar{X}_g)^\prime (X_{gt}-\bar{X}_g)=\sum_{g\in\mathcal{G}}\sum_{t=1}^T(X_{gt}^\prime X_{gt}-\bar{X}_g^\prime\bar{X}_g)$.  To see why this is the case, note that we can write $\bar{X}_g$ as: $M_gX_g$, where $M_g=I_g-1_g(1_g^\prime1_g)^{-1}1^\prime_g$ is a group-specific demeaning operator, and $1_g$ a matrix which indicates membership to group $g$ as a column of ones when the observation belongs to the group, and 0s otherwise.  Note also that $M_g$ is an idempotent matrix. 
Then, $\dot{X}_{gt}^\prime \dot{X}_{gt}=(X_{gt}-\bar{X}_g)^\prime (X_{gt}-\bar{X}_g)=[(I-M_g)X_{gt}]^\prime(I-M_g)X_{gt}$, and the matrix $(I-M_g)$ is symmetric and idempotent.  Thus, the preceding quantity can be written as $X_{gt}^\prime(I-M_g)X_{gt}=X_{gt}^\prime X_{gt}-\bar{X}_g^\prime\bar{X}_g$ as required.   Also note, that the $K\times K$ matrix $\bar{X}_g^\prime\bar{X}_g$ can be generated from an underlying $K\times 1$ vector of group level means and the number of observations in each group.  Specifically, refer to a group level vector of variable means as $\bar{x}_g \equiv (\bar{x}_{1g} \ \bar{x}_{2g} \ \cdots \ \bar{x}_{Kg})$.  Then $\bar{X}_g^\prime\bar{X}_g=N_g\times\bar{x}_g^\prime \bar{x}_g$, where $N_g$ is the number of observations in group $g$. Identical logic holds to show that $\dot{X}_{gt}\dot{y}_{gt}=(X_{gt}^\prime y_{gt}-\bar{X}_g^\prime\bar{y}_g)$, and $\bar{X}_g^\prime\bar{y}_g$ can be generated from group level averages  $\bar{x}_g$, a $K\times 1$ vector, and scalar $\bar{y}_g$. 

Given this, implementing fixed effect models using grouped data generated in a cumulative fashion is a straightforward extension of Algorithm \ref{alg:GCLS}.   For ease of exposition we define $\dot{X}^\prime_{g}\dot{X}_{g}\equiv\sum_{t=1}^T\dot{X}^\prime_{gt}\dot{X}_{gt}$, and $\dot{X}^\prime_{g}\dot{y}_{g}\equiv\sum_{t=1}^T\dot{X}^\prime_{gt}\dot{y}_{gt}$. From \eqref{eqn:grouped}, elements $X^\prime_{g}X_{g}$ and $X^\prime_{g}y_{g}$ have already been calculated cumulatively for all $g$.  The remaining step is to calculate $\bar{X}_g^\prime\bar{X}_g$ and $\bar{X}_g^\prime\bar{y}_g$ which only requires group-level variable means.  From this, $K\times K$ matrices $\dot{X}_{g}^\prime \dot{X}_{g}$ and $K\times1$ vector $\dot{X}_{g}^\prime \dot{y}_{g}$ can be generated, and the fixed effect estimator \eqref{eqn:FE} can be calculated as:
\[
\widehat\beta_{FE}=\left(\sum_{g\in\mathcal{G}}\dot{X}_{g}^\prime \dot{X}_{g}\right)^{-1}\left(\sum_{g\in\mathcal{G}}\dot{X}_{g}\dot{y}_{g}\right)
\]
Similar cumulative procedures can be followed for two-way fixed effect models using the double within-transformation \cite{Baltagi2001,Wooldridge2021}.  For example for balanced panels over group $g$ and time $t$, two way transformations $\ddot{X}_{gt}=X_{gt}-\bar{X}_g-\bar{X}_t+\bar{X}$ and $\ddot{y}_{gt}=y_{gt}-\bar{y}_g-\bar{y}_t+\bar{y}$ can be calculated, and similar procedures followed as in the fixed effect case.\footnote{This results follows from the case of single demeaning.  However, here both group and time fixed effects need to be removed. Noting that we can now define the double-demeaning operation as $M_{gt}=[I_{gt}-1_g(1_g^\prime 1_g)^{-1}1_g^\prime-1_t(1_t^\prime 1)^{-1}1^\prime+1(1^\prime 1)^{-1}1^\prime]$, and hence write $\ddot{X}$ as $M_{gt}X_{gt}$, then $\ddot{X}^\prime\ddot{X}=X_{gt}^\prime M^\prime_{gt}M_{gt}X_{gt}$.  However, $M_{gt}$ is idempotent, and so $\ddot{X}^\prime\ddot{X}=(X_{gt}^\prime X_{gt}-\bar{X}^\prime_{g}\bar{X_g}-\bar{X}^\prime_{g}\bar{X_g}+\bar{X}^\prime\bar{X})$.  This then suggests a simple and feasible process for concentrating out two-way (or higher order) fixed effects by grouping aggregates $\Sigma_{1\sim J}$ and $\Upsilon_{1\sim J}$ over $g$ and $t$, and calculating group-specific, time-specific, and overall means, which can then be used to estimate $\widehat\beta_{FE}$ after processing all data.}

\paragraph{Returning Fixed Effects}
Generally, when fixed effect models are implemented, the interest is in estimating the coefficients and standard errors on time-varying variables, and hence a fixed effect estimator like \eqref{eqn:FE}
is appropriate.  However, in cases where estimates and standard errors on fixed effects themselves are also desired, cumulative least squares procedures offer a particularly efficient way to generate these estimates.

To see this, note that in the case of mutually exclusive fixed effects, we can write:
\[
X^\prime X = \left( 
  \begin{array}{cccccc}
  \sum_{i=1}^Nx_{1i}x_{1i} & \cdots & \sum_{i=1}^Nx_{1i}x_{Ki} & 0 & \cdots & 0  \\
  \vdots & \ddots & \vdots & \vdots & \ddots & \vdots  \\
  \sum_{i=1}^Nx_{Ki}x_{1i} & \cdots & \sum_{i=1}^Nx_{Ki}x_{Ki} & 0 & \cdots & 0  \\
  \bar{x}_{1,g_{1}} & \cdots & \bar{x}_{K,g_{1}} & N_{g_1}& \cdots & 0  \\
  \vdots & \ddots & \vdots & \vdots & \ddots & \vdots  \\
  \bar{x}_{1,g_{K}} & \cdots & \bar{x}_{K,g_{K}} & 0 & \cdots & N_{g_K}  \\
  \end{array}\right)
\qquad
X^\prime y = \left( 
  \begin{array}{cccccc}
  \sum_{i=1}^Nx_{1i}y_{i} \\
  \vdots \\
  \sum_{i=1}^Nx_{Ki}y_{i} \\
  \bar{y}_{g_{1}}\\
  \vdots  \\
  \bar{y}_{g_{K}}\\
  \end{array}\right),
\]
where here we assume data is ordered such that first time-varying variables are included in $X$, and then group fixed effects.  In this case, the resulting matrix $X^\prime X$ simply consists of the $K\times K$ matrix $\Sigma_{1\sim J}$ in the top-left corner (where here $K$ refers to time-varying variables, a $N_G\times K$ matrix of group means in the bottom left corner, the matrix $O_{K,N_G}$ in the top right-hand corner, and a $N_G\times N_G$ diagonal matrix containing the number of observations in each group on the main diagonal.  Similarly, $X^\prime y$ simply consists of the vector $\Upsilon_{1\sim J}$ in positions 1 to $K$, and then $N_G$ group-level means below. In this case, the only required information beyond elements already stored in standard cumulative procedures ($\Sigma_{1\sim J}$ and $\Upsilon_{1\sim J}$), are group level means and observation numbers, which can be trivially estimated cumulatively.  This thus suggests that fixed effect estimators can be estimated directly and efficiently \emph{including all fixed effects} in a sequential procedure.

\paragraph{Cross-Validation}
In Section \ref{sscn:alternative} we noted that cumulative procedures could be used for models such as Ridge, LASSO and elastic net, where commonly tuning parameters are chosen.  Often, such tuning parameters are chosen through $k$-fold cross-validation (see, \emph{eg} \citeN{WuWang2020}).  We showed previously that the tuning parameter $\lambda$ in these models can be chosen after accumulating matrices $X^\prime X$ and $X^\prime y$ (see for example the case of Ridge regression in \eqref{eqn:ridge}).  If we follow Algorithm \ref{alg:GCLS}, where the group variable is simply a discrete uniform random variable taking values between 1 and $k$, resulting matrices $X_1^\prime X_1, \cdots, X_K^\prime X_K$, and $X_1^\prime y_1, \cdots, X_K^\prime y_K$, can be used for $k$-fold cross validation in an efficient way.

To see this, note that cross validation consists of a procedure where for a tuning parameter $\lambda$, a specific group $g$ is held out, and coefficients $\widehat\beta_{-g,Ridge}(\lambda)$ estimated using the remaining groups.  Within group $g$, the Mean Squared Error associated with this parameter is then calculated as $MSE=\frac{1}{N_{g}}||y_g-X_g\widehat\beta_{-g,Ridge}(\lambda)||^2$.  A similar procedure is then conducted for each of the $N_G$ groups, and the MSE associated with $\lambda$ is calculated as the sum of the group-specific MSEs.  Note that this quantity $\frac{1}{N_{g}}||y_g-X_g\widehat\beta_{-g,Ridge}(\lambda)||^2$ can be rewritten as: 
\begin{align*}
&\frac{1}{N_{g}}\left[(y_g-X_g\widehat\beta_{-g,Ridge}(\lambda))^\prime(y_g-X_g\widehat\beta_{-g,Ridge}(\lambda))\right]\\
&=y_g^\prime y_g-2\widehat\beta_{-g,Ridge}(\lambda)^\prime X_g^\prime y_g+\widehat\beta_{-g,Ridge}(\lambda)^\prime X_g^\prime X_g\widehat\beta_{-g,Ridge}(\lambda)
\end{align*}
Each of the quantities $X_g^\prime y_g$, $X_g^\prime X_g$ and $y_g^\prime y_g$ are already calculated in a cumulative fashion, implying that the MSE for a given lambda can be calculated entirely from cumulatively calculated aggregates, and MSE-optimal tuning parameters chosen as the value of $\lambda$ which minimises this MSE.

\subsection{Alternative Inference Procedures}
\label{sscn:inference}
In Section \ref{sscn:COLS} we documented that inference could be conducted in a cumulative fashion in the same way as point estimates, and this required no other special procedures, apart from the accumulation of $\Psi_{1\sim J}$, which is needed to calculate the variance-covariance matrix but not parameter estimates.  This can all be done in a single pass through blocks of the data.  
However, this relies on a homoscedasticity assumption. Here we discuss how inference can proceed in alternative settings.

\subsubsection{Heteroscedasticity Robust Standard Errors}
In cases where heteroscedasticity-robust standard errors are desired, the well-known heteroscedasticity-robust estimator can be implemented cumulatively.  The HC1 variance estimator for OLS is written as:
\[
\widehat{V}(\widehat\beta_{OLS})_{HC1} = \frac{N}{N-K}(X^\prime X)^{-1}\left[\sum_{i=1}^N\widehat{u}_i^2x^\prime_ix_i\right](X^\prime X)^{-1}
\]
From Section \ref{sscn:COLS}, we already know that $X^\prime X=\Sigma_{1\sim J}$ can be generated cumulatively.  Similarly, both $K$ and $N$ can be read trivially from data.  If $\widehat{u}_i^2$ is known, the central component $\sum_{i=1}^N\widehat{u}_i^2x^\prime_ix_i$ could be calculated cumulatively: this value, which we will refer to as $\Omega$ could be initialised as a null matrix $O_{K,K}$, and in each block of the dataset when an observation $i$ is read in, the quantity $\widehat{u}_i^2x^\prime_ix_i$ calculated, and added to all previous values, as laid out in the Algorithm \ref{alg:HC1} below.  As above, we will define $\Omega_{j}\equiv\sum_{i\in j}\widehat{u}_i^2x_i^\prime x_i$, and $\Omega_{1\sim j}\equiv \Omega_1+\cdots\Omega_j$.

The issue here however is that when the data is first loaded in blocks, we cannot calculate $\widehat{u}_i=(y_i-X_i\widehat{\beta}_{OLS})$, as this requires $\widehat{\beta}_{OLS}$, which is not known until an entire pass through the data has been completed.  Thus, while heteroskedasticity robust estimates can be calculated in a cumulative fashion, this requires the data be read in a cumulative fashion a second time.  In particular, first Algorithm \ref{alg:CLS} should be run to calculate $\widehat\beta_{OLS}$, and then Algorithm \ref{alg:HC1} be run with $\widehat\beta_{OLS}$ as an input. However, apart from having to return to read the data, there is no particular memory restriction which implies that this procedure will not be feasible.  The only addition is a single accumulated $K\times K$ matrix $\Omega_{1\sim J}$. Similar procedures can be conducted for IV and other estimators.

\begin{algorithm}
\caption{\small Cumulative Estimation of Heteroscedasticity-Robust Variance-Covariance Matrix}
\label{alg:HC1}
Inputs: Database consisting of ($y$,$X$), block size $b$. Point estimate $\widehat\beta_{OLS}$, and $\Sigma_{1\sim J}$ from Algorithm \ref{alg:CLS}.\\
Result: Variance-covariance matrix $\widehat{V}(\widehat\beta_{OLS})_{HC1}$.  \\ \ \\
1. Set $i=1$ and $j=b$. Load into memory partition of data covering $y,X$ in observations $i$ to $j$ \;
2. Calculate $\Omega_1$ \;
3. If observations $i$ to $j$ contain end of file, set $e=1$, otherwise, set $e=0$ \;
\While{$e\neq1$}
{
4. Replace $i=i+b$ and $j=j+b$. Load into memory partition of data covering $y,X$ in observations $i$ to $j$. \;
5. Calculate $\Omega_j$. \;
6. Calculate $\Omega_{1\sim j}=\Omega_{1\sim j-1}+\Omega_j$. \;
7. If observations $i$ to $j$ contain end of file, set $e=1$. \;
}\textbf{end}\\ \ \\
8. Calculate: $\widehat{V}(\widehat\beta_{OLS})_{HC1} = \frac{N}{N-K}\Sigma_{1\sim J}^{-1}\Omega_{1\sim J}\Sigma_{1\sim J}^{-1}.$
\end{algorithm}




\subsubsection{Cluster-Robust Variance Covariance Matrix}
\label{sscn:clusterVC}
Similarly, in the case of standard closed-form cluster-robust variance-covariance estimators, a second pass through of the data is required to calculate cumulative standard errors.\footnote{In Section \ref{sscn:bootstrap} we lay out an extremely efficient clustered bootstrap procedure in which it is not necessary to return to individual-level data.}  In this case, slightly more information must be stored, namely an additional vector of size $K$ for each of the $N_G$ groups over which clustering occurs, but unless both $K$ and $N_G$ are exceedingly large, this should not generate a problem for the feasibility of these procedures.  Perhaps somewhat surprisingly, while clustered variance-covariance matrices account for dependence among observations, it is never necessary for data for an entire cluster to be housed in a computer's working memory in order to cluster standard errors by group.


To see this, note that the standard cluster-robust variance-covariance estimator is written as follows.
\[
\widehat{V}(\widehat\beta_{OLS})_c = \frac{N-1}{N-K}\frac{N_G}{N_G-1}(X^\prime X)^{-1}\left[\sum_{g\in\mathcal{G}}X^\prime_g(\widehat{u}_g\widehat{u}^\prime_g)X_g\right](X^\prime X)^{-1}
\]
As previously $g$ refers to groups over which clustered standard errors are desired, and $N_G$ refers to the total number of groups. As in the case of the HC1 estimator, observation, group, and covariate quantities can be easily read in a cumulative fashion from data, and $X^\prime X$ is similarly calculated cumulatively.  However, here we additionally require the quantity $\Omega_g\equiv\sum_{g\in\mathcal{G}}X^\prime_g(\widehat{u}_g\widehat{u}^\prime_g)X_g$. For expositional clarity, note that $\sum_{g\in\mathcal{G}}X^\prime_g(\widehat{u}_g\widehat{u}^\prime_g)X_g=\sum_{g\in\mathcal{G}}(X^\prime_g\widehat{u}_g)(\widehat{u}^\prime_gX_g)$.  Matrix $X_g^\prime$ is an $K\times N_g$ matrix, while $\widehat{u}_g$ is an $N_g\times 1$ vector of regression residuals.  Thus, the matrix $X^\prime_g\widehat{u}_g$ is a $K\times 1$ vector, while its transpose $\widehat{u}^\prime_gX_g$ is $1\times K$.  Additionally, note that $X^\prime_g\widehat{u}_g$ is generated by multiplying the observations \emph{of each observation} with its own residual, and so can be generated cumulatively. Thus, as previously, if data is arbitrarily divided into $J$ blocks, the quantity $X^\prime_g\widehat{u}_g=X^{1\prime}_g\widehat{u}^1_g+\ldots + X^{J\prime}_g\widehat{u}^J_g$ can be generated cumulatively by first calculating  $X^{j\prime}_g\widehat{u}^j_g$ for each group present within each block $j$, then summing over all $j$, and finally using this quantity to calculate the overall quantity $\Omega_g$.\footnote{It is important to note that this procedure requires working with the $K\times 1$ vector $X^{j\prime}_g\widehat{u}^j_g$ at each step.  It is \emph{not} possible to calculate $\Omega_g$ at each step, but rather, we must accumulate $X^{j\prime}_g\widehat{u}^j_g$ and only then calculate $\Omega_g$.} This procedure is laid out formally in Algorithm \ref{alg:cluster} below, where as before, $X^{\prime}_g\widehat{u}_{g,1\sim j}$ refers to the summation of  $X^{1\prime}_g\widehat{u}^1_g+\cdots+X^{j\prime}_g\widehat{u}^j_g$.


\begin{algorithm}[ht!]
\caption{Cumulative Estimation of Cluster-Robust Variance-Covariance Matrix}
\label{alg:cluster}
Inputs: Database consisting of ($X, G$), block size $b$. Point estimate $\widehat\beta_{OLS}$, and $\Sigma^{g}_{1\sim J}$ from Algorithm \ref{alg:GCLS}.\\
Result: Variance-covariance matrix $\widehat{V}(\widehat\beta_{OLS})_{C}$.  \\ \ \\
1. Set $i=1$ and $j=b$. Load into memory partition of data covering $G,X$ in observations $i$ to $j$ \;
\For{$g\in\mathcal{G}^{1}$}
{
2. Calculate $X^{1\prime}_g\widehat{u}^1_g$ \;
}
3. If observations $i$ to $j$ contain end of file, set $e=1$, otherwise, set $e=0$ \;
\While{$e\neq1$}
{
4. Replace $i=i+b$ and $j=j+b$. Load into memory partition of data covering $X,G$ in observations $i$ to $j$. \;
\For{$g\in\mathcal{G}^{j}$}
{5. Calculate $X^{j\prime}_g\widehat{u}^j_g$. \;
6. If $\exists\ X^\prime_g\widehat{u}_{g,1\sim j-1}$ calculate $X^\prime_g\widehat{u}_{g,1\sim j}=X^\prime_g\widehat{u}_{g,1\sim j-1}+X^{j\prime}_g\widehat{u}^j_g$, otherwise initialise $X^\prime_g\widehat{u}_{g,1\sim j}=X^{j\prime}_g\widehat{u}^j_g$. \;
}
7. If observations $i$ to $j$ contain end of file, set $e=1$. \;
}\textbf{end}\\ \ \\
8. Calculate
\[
\widehat{V}(\widehat\beta_{OLS})_c = \frac{N-1}{N-K}\frac{N_G}{N_G-1}(\Sigma_{1\sim J})^{-1}\left[\sum_{g\in\mathcal{G}}(X^\prime_g\widehat{u}_{g,1\sim J})(X^\prime_g\widehat{u}_{g,1\sim j})^\prime\right](\Sigma_{1\sim J})^{-1}
\].

\end{algorithm}



\subsubsection{An Efficient Bootstrap Algorithm for Clustering}
\label{sscn:bootstrap}
While the clustered procedure described in the previous sub-section is feasible and permits for the exact calculation of analytic cluster-robust variance covariance matrices, it requires opening the data two times: the first to calculate the parameter estimates, and the second to calculate the standard errors which requires residuals $\widehat{u}_g$.  However, given the results from Section \ref{sscn:groups}, if one wishes to generate a clustered standard error by bootstrapping, this can be done in a single pass through the data, and additionally bootstrap replicates can be conducted extremely quickly, and indeed orders of magnitude more quickly than in standard clustered bootstrap procedures.  To see why, note that the parameter estimate of interest can be generated as in \eqref{eqn:grouped}.  Also note that from Algorithm \ref{alg:GCLS}, that cumulative procedures are used to generate $K\times K$ matrices for each group $\Sigma_{1\sim J}^{g}$, as well as group-specific $K\times 1$ vectors $\Upsilon_{1\sim J}^{g}$.

This implies that we can generate resampled versions of \eqref{eqn:grouped} by simply resampling with replacement $N_G$ pairs of matrices $\Sigma_{1\sim J}^{g},\Upsilon_{1\sim J}^{g}$, and calculating a resampled estimator $\widehat{\beta}^{b*}$ as follows:
\begin{equation}
\label{eqn:bootstrap}
\widehat\beta^{*} = \left(\sum_{g^*\in\mathcal{G}^*}\Sigma_{1\sim J}^{g*}\right)^{-1}\left(\sum_{g^{*}\in\mathcal{G}^{*}}\Upsilon_{1\sim J}^{g*}\right),
\end{equation}
where $\Sigma_{1\sim J}^{g*}$ refers to resampled matrix $\Sigma_{1\sim J}^{g}$, and similarly for $\Upsilon_{1\sim J}^{g}$.  When clusters are large, such as individuals within states of countries, resampling aggregated matrices to form bootstrap resamples $\widehat\beta^{*}$ will be orders of magnitudes faster than resampling clusters of data. This suggests a potentially substantially faster bootstrap estimate for the cluster robust variance for the parameter vector $\widehat\beta$.  This consists of generating a large number $B$ of resampled estimates  
\eqref{eqn:bootstrap}, which can be used to calculate the bootstrap CRVE for $\widehat\beta$ as: $V(\widehat\beta)_{CRVE}=V(\widehat\beta^{*})=\frac{1}{B}\sum_{b=1}^B(\widehat\beta_b^{*}-E[\widehat\beta_b^{*}])^2$.



\section{Optimal Implementation}
\label{scn:optimal}
Whether implementing cumulative or standard algorithms, identical calculations are required to be made, given that cross products are required between each element of $X$ and between $X$ and $y$ for each observation $i$.  Indeed, cumulative algorithms require \emph{strictly more} calculations than standard algorithms.  To see this, consider the case of OLS.  To calculate coefficients in OLS, $X^\prime$ must be multiplied with $X$, implying computational time of order $\mathcal{O}(NK^2)$.  Additionally, $X^\prime$ must be multiplied with $y$, implying computational time of order $\mathcal{O}(NK)$. Finally, resolving the linear system $X'X \beta = X'y$ involves time $\mathcal{O}(K^3)$ via Gauss-Jordan elimination.  In the case of cumulative algorithms, identical procedures are required, and additionally, at each step two $K\times K$ matrices $\Sigma_j$ and $\Sigma_{1\sim j-1}$ must be summed, which is of computational time $\mathcal{O}(K^2)$, and similarly, two $K\times 1$ vectors $\Sigma_j$ and $\Sigma_{1\sim j-1}$ must be summed, involving time $\mathcal{O}(K)$.  In general, $N>>K$, implying that $\mathcal{O}(NK^2)$ will dominate in both cases.  

Nevertheless, if all computational procedures scale linearly in the number of observations, no gains will be made by implementing cumulative routines in place of their standard counterparts.  However, computational routines clearly do not scale linearly with sample size indefinitely.  To see this, it is sufficient to consider two cases: one where $N$ is such that observations can be housed in a computer's working memory, and another where $N$ exceeds the capacity of a computer's memory.  In the prior case, the calculation time will be finite, while in the latter case calculation will be impossible, and hence time will be infinite.  In this section we will discuss the optimal implementation where optimality refers to the block size which minimises calculation time.  Given that in the limit cumulative procedures simply revert to standard OLS estimation if a block size of $N$ is chosen, we consider only the optimal choice of block size for cumulative procedures. We return to these issues empirically in Section \ref{scn:illustrations}.

As above, the entire cumulative algorithm for OLS requires a number of well-defined steps. In total, matrix multiplication between $X'$ and $X$ is $\mathcal{O}(NK^2)$, and between $X'$ and $Y$ is $\mathcal{O}(NK)$.  Final resolution of the parameters is $\mathcal{O}(K^3)$.  Additionally, within each block $j$ a series of element-by-element summations must occur to accumulate $\Sigma_{1\sim j}$ and $\Upsilon_{1\sim j}$. In each step these are are of order $\mathcal{O}(K^2)$ and $\mathcal{O}(K)$ respectively.  Given that there are $J$ such blocks, and in the first block it is not necessary to accumulate $\Sigma_{1\sim 1}$ and $\Upsilon_{1\sim 1}$, these calculations are of computational time $\mathcal{O}(K^2(J-1))$ and $\mathcal{O}(K(J-1))$.  Thus, total computational time of the algorithm is of the order:  
\begin{equation}
\label{eqn:complexity}
\mathcal{O}(NK^2)+\mathcal{O}(NK)+\mathcal{O}(K^3)+\mathcal{O}(K^2(J-1))+\mathcal{O}(K(J-1)).
\end{equation}
Here it is clear that if a single block is chosen, and hence $J=1$, then $\mathcal{O}(K^2(J-1))+\mathcal{O}(K(J-1))=0$ and the cumulative algorithm collapses to OLS.

To consider the optimal block size, we will consider separately three elements of \eqref{eqn:complexity}.  A first element, corresponding to the first two terms in \eqref{eqn:complexity} and denoted $L(N,K)$ is the procedure of loading data and multiplying matrices required to arrive to $X'X$ and $X'y$.  We write this function as $L(N,K)=l(NK^2+NK)$.  A second element, corresponding to the third term in \eqref{eqn:complexity}, consists of generating estimates $\widehat\beta$ once provided with $X'X$ and $X'y$, and is written as $S(K)=s(K^3)$.  And finally, an accumulation procedure, denoted $C(J,K)=c(K^2(J-1)+K(J-1))$, consisting of the final two terms where matrices are summed in a cumulative fashion.

Note that given that $N=JN_j$, the first term can be re-expressed as s $L(J,K)=l(JN_jK^2+JN_jK)$.  For the sake of simplicity, given that $N_j$ is determined by $J$, below we omit $N_j$ terms as implicit in $l(\cdot)$.  For a given $K$, The total time to compute the cumulative least squares algorithm can thus be written as:
\[
T_c(J;K) = L(J,K)+S(K)+C(J,K).
\]
Hence, the optimal number of partitions of data $J$ should solve the problem:
\[
\underset{J}{\min} \bigg(L(J,K) + S(K) + C(J,K)\bigg) \text{ subject to }0 < J \leq N.
\]
For an interior solution, this suggests that optimal number of blocks considered should satisfy the following first order condition:
\begin{align}
&\frac{\partial L(J,K)}{\partial J}+\frac{\partial C(J,K)}{\partial J}=0 \nonumber \\
\Rightarrow &\frac{\partial L(J,K)}{\partial J} =-\frac{\partial C(J,K)}{\partial J} \label{eqn:optimal}
\end{align}
Note that here given that regardless of the block size chosen, the same final matrix inversion is required, for a given $K$ optimality does not depend on $S(\cdot)$, as reflected in \eqref{eqn:optimal}.
This suggests the logical conclusion that an optimal block size should be chosen which equates the marginal cost of summing an additional set of matrices across blocks, $\partial C(J,K)/\partial J$, with the marginal benefit coming from loading smaller partitions of the data into memory to calculate $X'X$ and $X'Y$.  

Understanding the optimal block size for conducting cumulative least squares thus requires understanding the nature of functions $C(J,K)$ and $L(J,K)$.  The precise nature of these two functions is likely highly dependent upon a particular computational environment (both software and hardware), nevertheless, we can suggest a number of key conjectures.  Firstly, it is clear that for a given $K$, $C(J,K)$ will, abstracting from other elements, be linear in $J$.  To see this, note that for each additional block, we simply require the summation of an additional identically sized $K\times K$ and  $K\times 1$ matrix.  Thus, moving from $j$ to $j+1$ requires adding one set of summations, while moving from $j+1$ to $j+2$ requires adding an identical set of summations, and so calculation time will scale linearly in block sizes.  Secondly, for a given $K$, $L(J,K)$ seems unlikely to be linear in $J$.  Rather, this value is highly dependent on a particular computational environment.  Note that in general, when a computer's RAM usage is high a number of internal processes such as paging occur such that loading data into memory becomes increasingly slow as the size of a database increases.  Thus, when a sample approaches the limit of a computer's RAM, the marginal benefit of increasing the number of blocks of data is high given that it avoids substantial slowdowns inherent in computational architecture.  However, if a computer's RAM usage is low, the marginal benefit of increasing the number of blocks approaches zero, given that no such slowdown in data loading occurs, and the total computation time $L(J,K)$ is independent of $J$.  Thus, at very high values of $J$, for example where $J$ approaches the total number of observations, the marginal benefit of increasing $L(J,K)$ is likely essentially zero given that no memory slowdown occurs owing to the storage of large amounts of data in memory. However, at high low values of $J$, if data is large enough to result in memory slowdowns, the benefit of increasing $J$ is substantial.  On the other hand, the costs of increasing $J$, $C(J,K)$ are constant in $J$.


This suggests a number of general results. Firstly, if one is working with large datasets and memory is not unlimited, it is likely the case that smaller blocks of data should be preferred given that memory slowdowns can be avoided.  If data does not fit in memory, this argument holds with certainty, given that $\partial L(J,K)/\partial J|_{J_{min}}=\infty$, where $J_{min}$ refers to the point at which it becomes feasible to hold data in memory.  However, if RAM limits are binding with $N$, the optimal solution is likely not to increase the number of blocks to the maximum theoretical limit ($J=N$), given that at low block sizes no memory slowdown will be observed, but a constant cost increase is observed in terms of sums across blocks.  What's more, these results suggest that there is no gain from varying the block size across the sample, but rather that a single value of $N_J$ should be chosen as that which satisfies \eqref{eqn:optimal}. Finally, as the number of covariates increases, it seems likely that fewer blocks should be preferred, given that the cost of adding marginal blocks increases in $K$.  Precise optima will vary across computers and configurations, and are thus specific to particular contexts.  In the following section we will document specific examples which point to a Goldilocks principle of choosing blocks neither too big nor too small, and, fortunately, suggest that computation time is quite flat over a large range of blocks, provided that extreme situations are not encountered.

%
%

\section{Illustrations}
\label{scn:illustrations}
In this section we document two examples to illustrate the performance of cumulative procedures in practice. A first example is based on simulated data where we maintain fixed computational resources and vary key parameters of the data (namely the number of observations and the number of variables).  And a second example is based on real data, where we document the performance of cumulative versus standard estimation procedures in a range of computational environments and with various methods of estimation.

\subsection{Simulated Data}
To demonstrate the relative performance of the cumulative algorithm compared with a standard regression implementation, we test the time to complete calculations under controlled conditions.  Specifically, we compare the time it takes to run an OLS regression using cumulative and standard estimation routines based on the same data.  We conduct these tests on a server with 1GB of dedicated RAM and no outside processes running to ensure comparability across estimation times.\footnote{This is a commercially available Virtual Private Server with a 4 core CPU running a Linux-based operating system. All data is stored on the server on a solid state drive.}  We consider a range of observation numbers and independent variables, and, in the case of the cumulative algorithm, also document times under a range of block sizes.  In each case, the time completed consists of identical procedures: namely, in the case of the cumulative algorithm it is the time to import all blocks of data, calculate the necessary block-specific quantities, and finally return the regression estimates, standard errors and R-squared.  And in the case of a `standard' regression implementation, the time simply refers to the time to open the data from the disk and estimate the OLS regression using canned software. 

The test procedure thus consists of generation of data of the following general form:
\[
y = X\beta + u,
\]
where $X$ is an $N\times K$ matrix of simulated data consisting of a constant and $K-1$ uniformly distributed variables, $u\sim\mathcal{N}(0,3)$ is a simulated $N\times1$ error term, and $\beta$ is a $K\times 1$ vector of parameters.  Here we consider processing times varying $K$, $N$ and the block size, $N_j$, where in each case $X$ and $y$ are treated as inputs, $u$ as unobservable, and $\beta$ as a vector of parameters to estimate.

Tests are conducted using a recent version of Stata (specifically, Stata version 16), where the cumulative algorithm is written principally in Stata's matrix language Mata.  Regression is conducted using Stata's native ``regress'' command.  Initially, a single core version of Stata is used (Stata SE), however relative performance is shown to follow qualitatively similar patterns when a multiple processor version of Stata is used (Stata MP).  In Section \ref{sscn:empirics} we consider a range of alternative estimation procedures and models. 



\begin{figure}[ht!]
  \caption{\small Sample size and execution time of commercial routine versus cumulative method}
  \label{fig:sampleSize}
  \begin{subfigure}{0.49\textwidth}
  \includegraphics[scale=0.52]{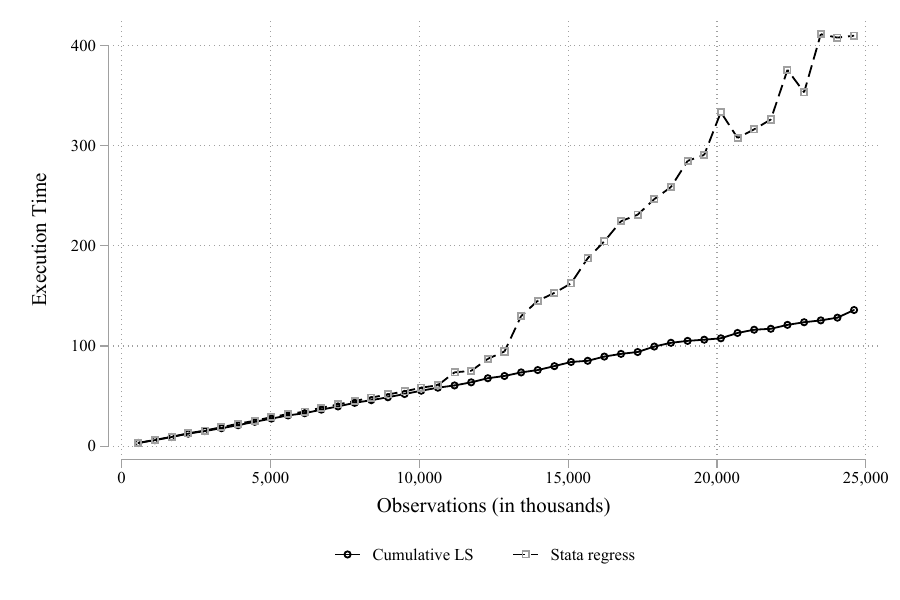}
  \caption{5 Independent Variables}
  \end{subfigure}
  \begin{subfigure}{0.49\textwidth}
  \includegraphics[scale=0.52]{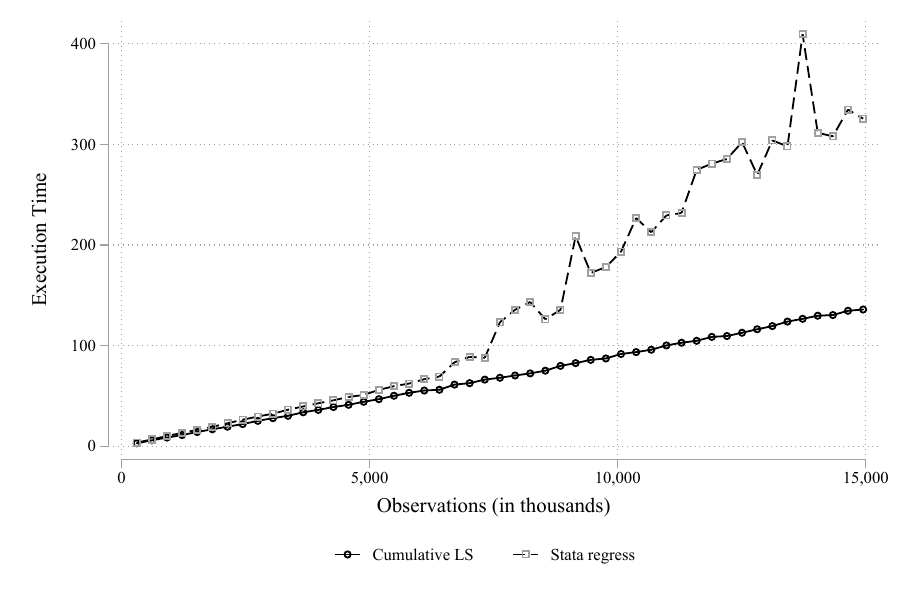}
  \caption{10 Independent Variables}
\end{subfigure}

  \begin{subfigure}{0.49\textwidth}
  \includegraphics[scale=0.52]{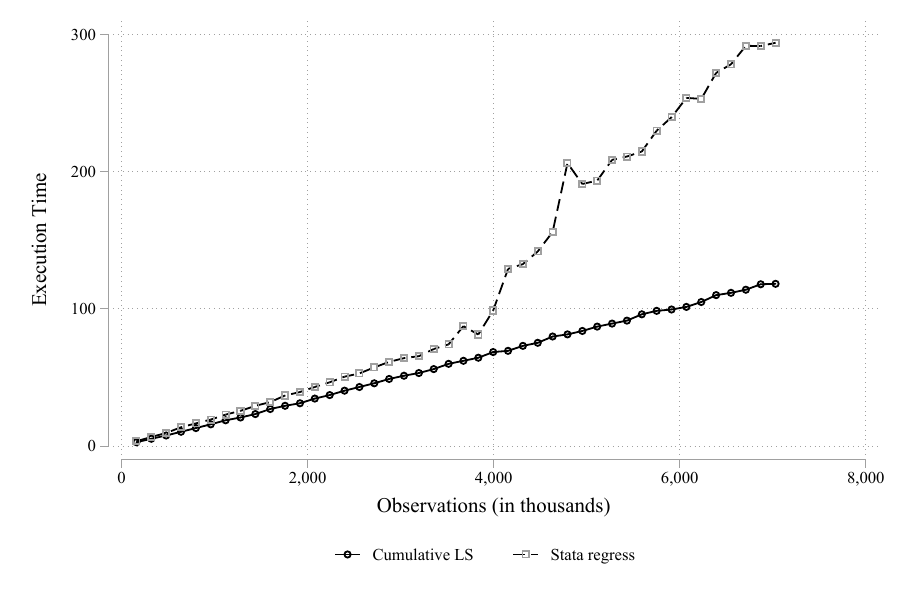}
  \caption{20 Independent Variables}
  \end{subfigure}
  \begin{subfigure}{0.49\textwidth}
  \includegraphics[scale=0.52]{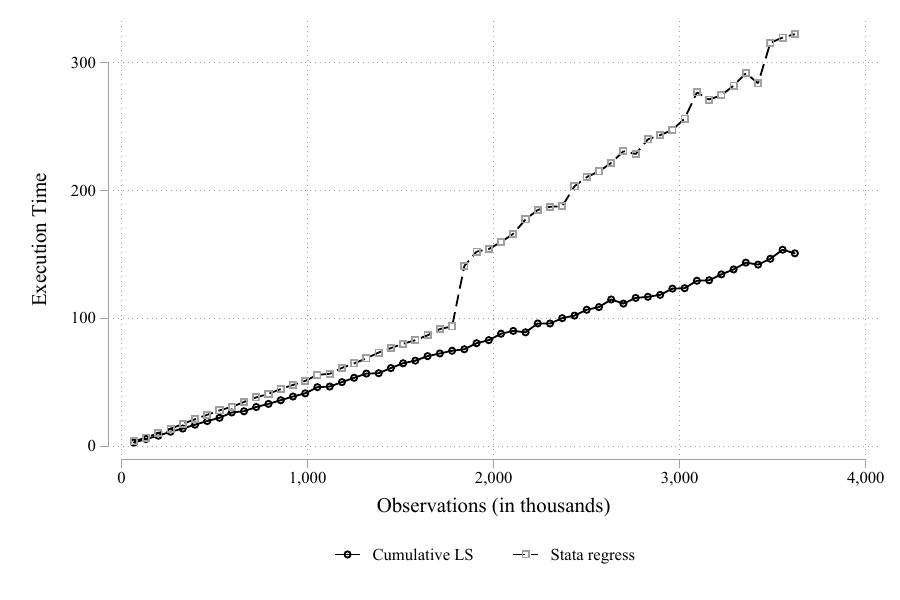}
  \caption{50 Independent Variables}
  \end{subfigure}
  \floatfoot{\textbf{Notes}: All times refer to the computation time of reading data into memory and estimating an ordinary least squares regression. Tests are all conducted on a system with 1GB of RAM, with no other processes running.  In each case, tests are conducted up to the point at which there is insufficient RAM to open the data, this precluding the estimation of standard regression models. Beyond this point, it is still feasible to estimate parameters using Cumulative Least Squares.}
\end{figure}

Processing times for estimation of cumulative algorithms versus standard regression software are documented in Figure \ref{fig:sampleSize}.  Each panel presents processing times for a particular number of simulated independent variables ranging from 5 (panel (a)), to 50 (panel (d)).  Processing time in seconds is documented on the vertical axis of each plot, and the total number of observations in thousands is documented on the horizontal access.  Times for standard regression software are presented as hollow squares with dashed lines, while times for cumulative algorithms are presented as hollow circles connected by a solid black line.  Each point refers to a specific simulated dataset and the time it takes to estimate parameters, standard errors, and other regression statistics with this data. In this Figure, in each case where cumulative algorithms are used the block size is arbitrarily chosen to contain 10\% of the total number of observations.

Across all panels we observe, unsurprisingly, that as the total number of observations grow for a fixed $K$, processing time increases. For cumulative algorithms, this processing time increases approximately linearly.  For example, for the case where $K$=5, regressions with 5, 10, 15 and 20 million observations take approximately 30, 60, 90 and 120 seconds to run.  This is observed in all panels.  Similar linear behaviour is observed in standard regression software when the observation numbers are moderate compared to the total RAM available.  However, the linear relationship breaks down and processing times become considerably slower from around the time that the total number of observations approaches around 50\% of the memory capacity of the computer.\footnote{In principle, a computer with 1GB of RAM contains 1$\times10^{9}\times\frac{1024^3}{1000^3}$ bytes of memory. A given line of data in our simulations consists of $K$ double precision variables which each occupy 8 bytes of memory ($K-1$ independent variables and the dependent variable).  Thus, one can calculate the theoretical maximum number of observations which could be held in memory as $\frac{1\times10^{9}}{K\times8}\times\frac{1024^3}{1000^3}$. In Figure \ref{fig:sampleSize} we observe that the kink in processing times for Stata's regress command occurs at around 12,000,000 observations, or around 44\% of the computer's theoretical maximum observations.} This implies that at relatively small numbers of observations compared to a computer's available memory, the processing time of cumulative procedures are similar to that of standard non-cumulative procedures, however cumulative procedures then rapidly become 2 to 3 times fasted than non-cumulative counterparts.  At some point, when the number of observations grows beyond the capacity of the RAM, non-cumulative procedures become infeasible to estimate, while the processing time of cumulative procedures continue to scale linearly indefinitely.  If similar tests are run using multiple processor versions of software, similar patterns are observed (Appendix Figure \ref{tab:processingTimesMP}).

\begin{figure}[ht]
  \centering
  \caption{Execution Time of Cumulative Least Squares by Block Size}
  \label{fig:block2d}
  \begin{subfigure}{0.49\textwidth}
  \includegraphics[scale=0.52]{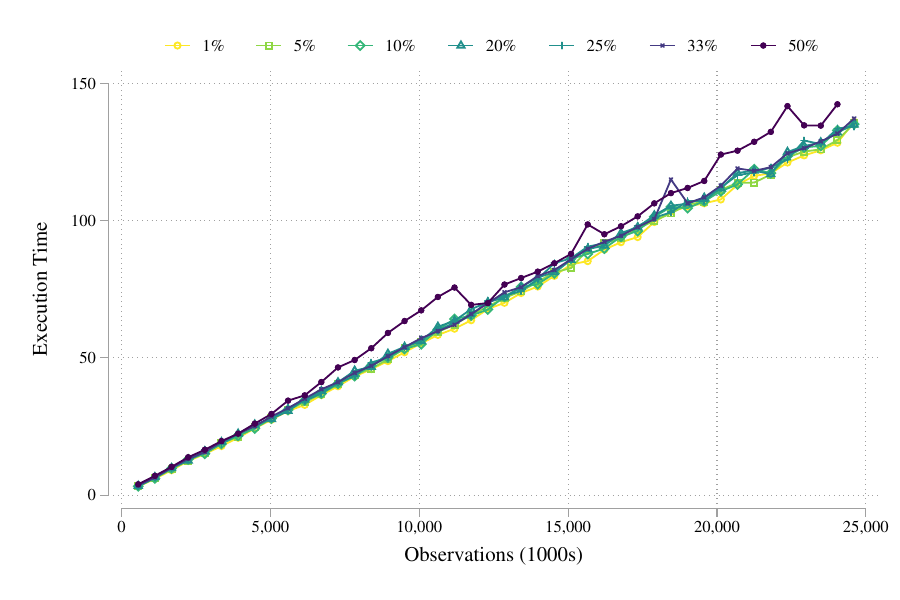}
  \caption{5 Independent Variables}
  \end{subfigure}
  \begin{subfigure}{0.49\textwidth}
  \includegraphics[scale=0.52]{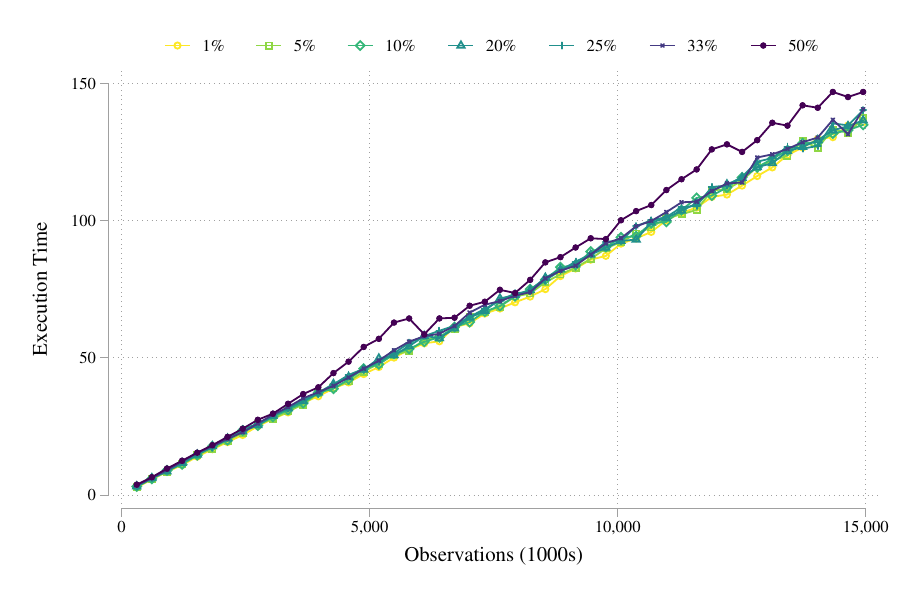}
  \caption{10 Independent Variables}
  \end{subfigure}

    \begin{subfigure}{0.49\textwidth}
  \includegraphics[scale=0.52]{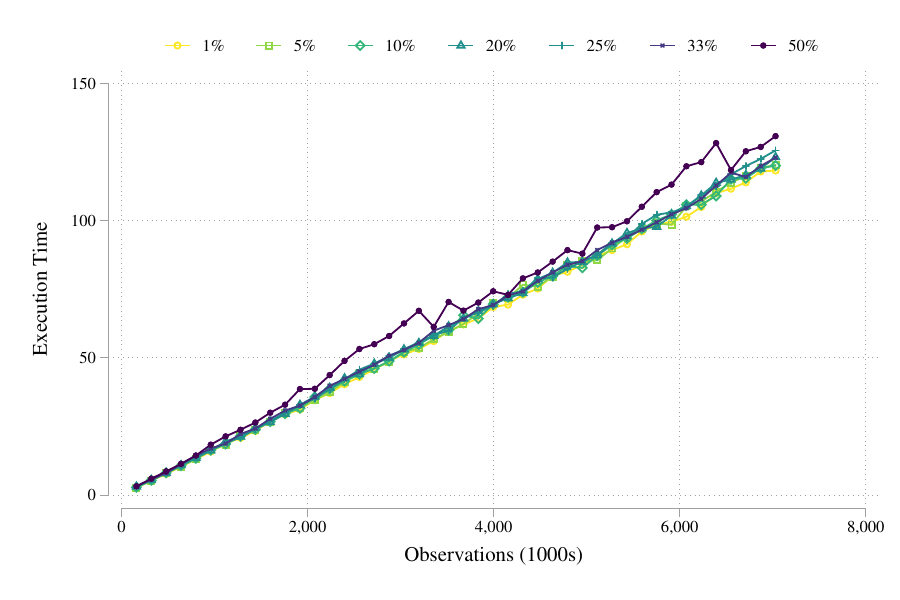}
  \caption{20 Independent Variables}
  \end{subfigure}
  \begin{subfigure}{0.49\textwidth}
  \includegraphics[scale=0.52]{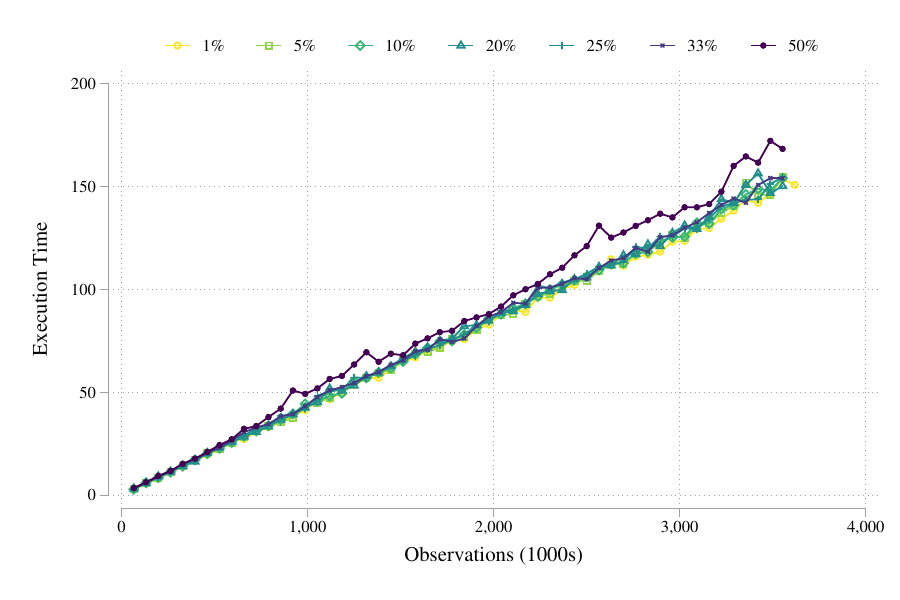}
  \caption{50 Independent Variables}
  \end{subfigure}
  \floatfoot{\textbf{Notes}: All times refer to the computation time of reading data into memory and estimating a cumulative least squares regression. Tests are all conducted on a system with 1GB of RAM, with no other processes running.  Block sizes as a proportion of the total observation numbers are indicated in the figure legend.}
\end{figure}

Results from Figure \ref{fig:sampleSize} are based on a block size $N_j$ which is arbitrarily chosen as $N_j=N/10$.  If data is very large, blocks of this size will also imply that individual blocks of data cannot fit in memory.  In Figure \ref{fig:block2d} we document processing times of cumulative least squares procedures where the block size is varied from 1\% of data up to 50\% of data (using sizes of greater than 50\% of data is not sensible, as one block will be larger than the other).\footnote{These are essentially profiles of a surface where the block size is varied continuously. The entire surface is plotted in Appendix Figure \ref{fig:block3d}.}  Once again, we document times across a range of values for $K$ (panels), and $N$ (horizontal axes). Each point refers to the time for a single regression.   We observe that across all cases examined, in general smaller block sizes are marginally faster.

\begin{figure}[ht]
  \centering
  \caption{\small Relative Time of Cumulative Least Squares Versus Standard Implementation by Block Size}
  \label{fig:block2db}
  \begin{subfigure}{0.49\textwidth}
  \includegraphics[scale=0.52]{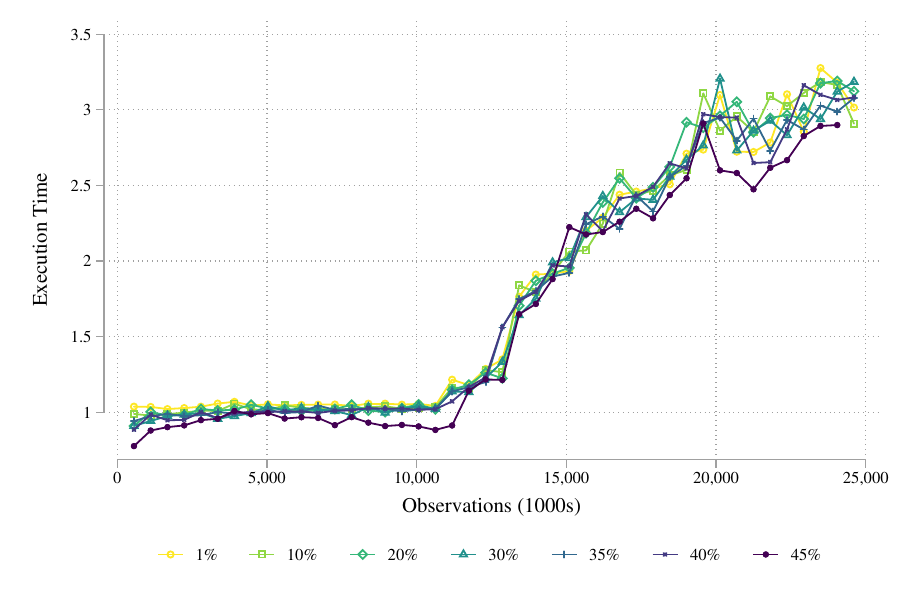}
  \caption{5 Independent Variables}
  \end{subfigure}
  \begin{subfigure}{0.49\textwidth}
  \includegraphics[scale=0.52]{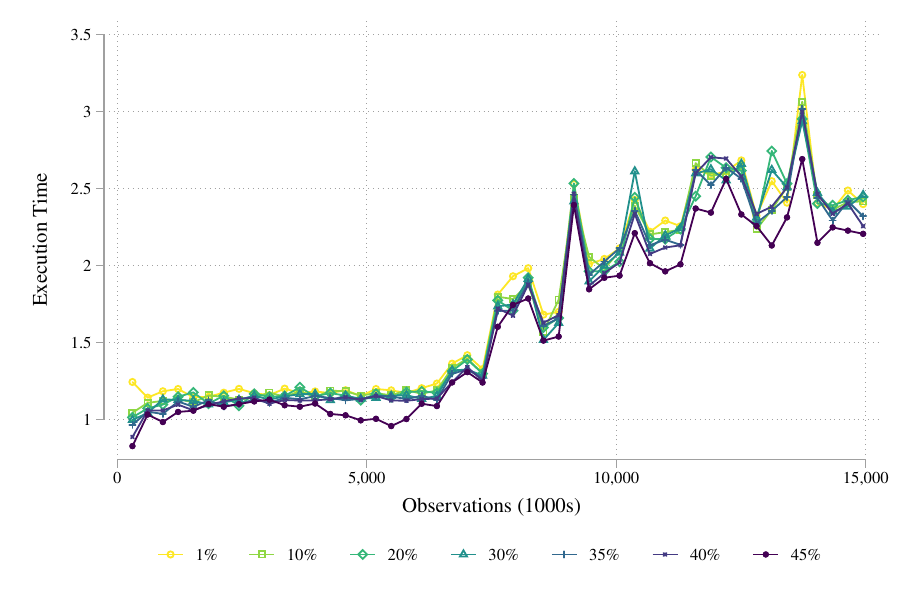}
  \caption{10 Independent Variables}
  \end{subfigure}

    \begin{subfigure}{0.49\textwidth}
  \includegraphics[scale=0.52]{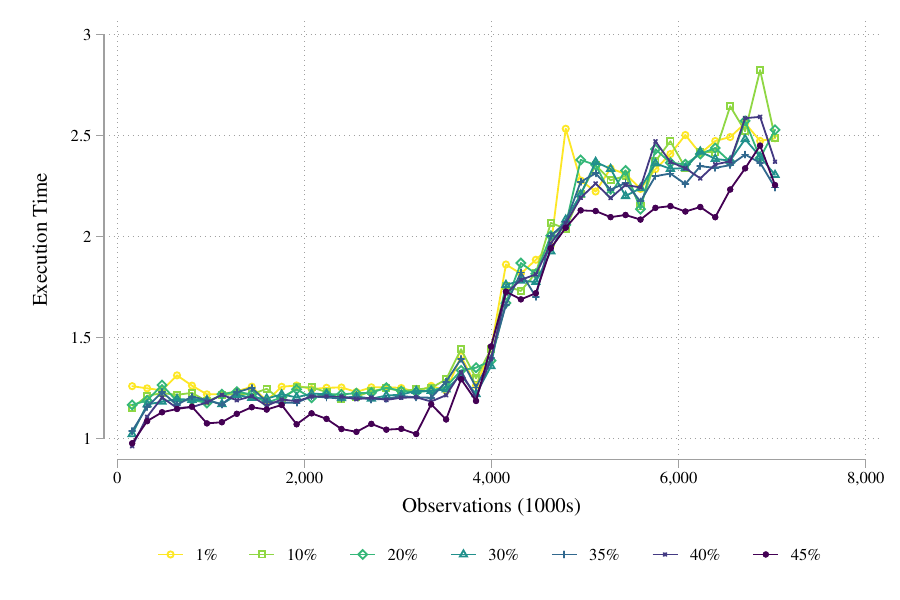}
  \caption{20 Independent Variables}
  \end{subfigure}
  \begin{subfigure}{0.49\textwidth}
  \includegraphics[scale=0.52]{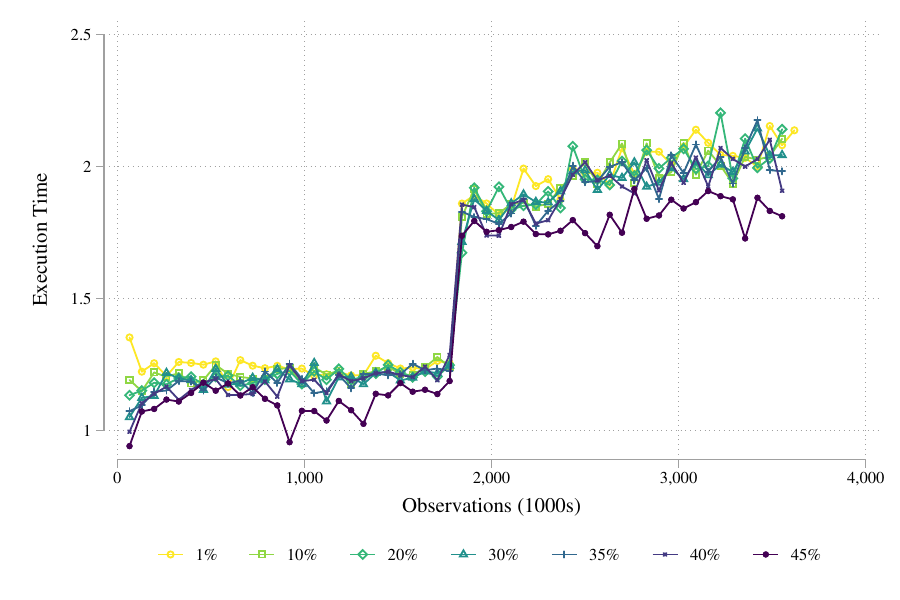}
  \caption{50 Independent Variables}
  \end{subfigure}
  \floatfoot{\textbf{Notes}: Each point presents the ratio of computation times of Stata's native regression command to cumulative least squares.  Values less than 1 imply non-cumulative procedures are faster than cumulative procedures, and vice versa for values greater than 1. All times refer to the computation time of reading data into memory and estimating a cumulative least squares regression. Tests are all conducted on a system with 1GB of RAM, with no other processes running.  Block sizes as a proportion of the total observation numbers are indicated in the figure legend.}
\end{figure}

Figure \ref{fig:block2db} documents the ratio of computation times from Stata's native regress command compared to cumulative least squares, where cumulative least squares is implemented with the same range of block sizes displayed in Figure \ref{fig:block2d}.  Values of less than 1 imply that standard (non-cumulative) estimation procedures are faster than cumulative procedures, while values greater than 1 imply that cumulative procedures are faster than non-cumulative procedures.  In line with the substantial increase in non-cumulative procedures documented in Figure \ref{fig:sampleSize}, we observe a sharp improvement in the ratio at around 50\% of the theoretical maximum memory.  In this particular implementation, when the smallest block size is used (1\% of $N$), the ratio is consistently greater than 1.

These results may suggest that the optimal procedure is thus to choose a block size as small as possible.  All results in this paper hold for block sizes as small as $N_j=1$, and even in cases where $N_j=N/100$, the block size considered exceeds 1.  In Figure \ref{fig:blockOpt} we consider execution times for a particular simulated dataset ($N$=25,000,000, $K=5$), however here allowing block sizes to fall to their smallest possible value.  A logarithmic scale is used on the horizontal scale allowing black sizes to vary from 1 to 12.5 million observations (50\% of the total observations).  Panel (a) uses an identical 1GB server as that used in tests above, while panel (b) documents the same times on a computer where memory limits do not bind.  In this case we observe that the optimal block size \emph{is not} the smallest possible size ($N=1$), but rather follows the Goldilocks principle laid out in Section \ref{scn:optimal}.  Clearly, block sizes that are so large that memory constraints begin to bind with $N_j$ observations should be avoided, however, a very small block size is also sub-optimal, given that this requires the accumulation of many $\XpXj{j}$ and $\Xpyj{j}$ matrices.  Where memory limits do not bind sharply (panel (b)) one may wish to work with slightly larger block sizes to avoid sub-optimal behaviour observed with very small block sizes, as provided extreme regions are avoided, the practical choice of block appears to be of second order importance.

\begin{figure}
  \centering
  \caption{Optimal Block Size and Available Memory} 
  \label{fig:blockOpt}
  \begin{subfigure}{0.49\textwidth}
  \includegraphics[scale=0.43]{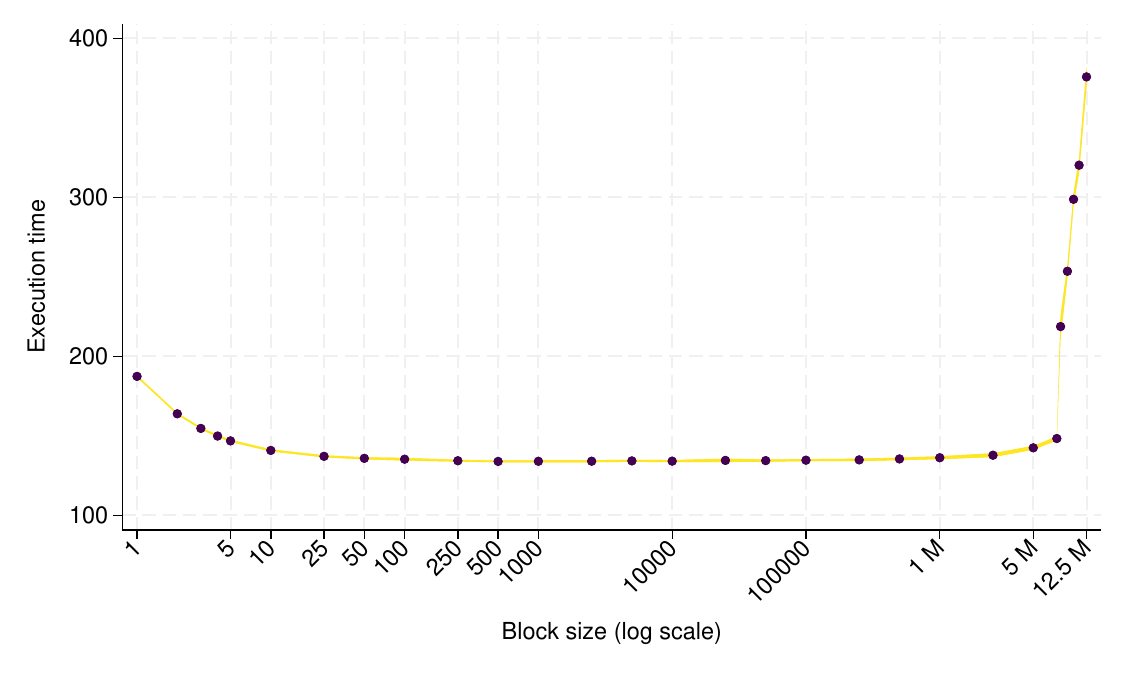}
  \caption{Memory Limits Binding}
  \end{subfigure}
  \begin{subfigure}{0.49\textwidth}
  \includegraphics[scale=0.42]{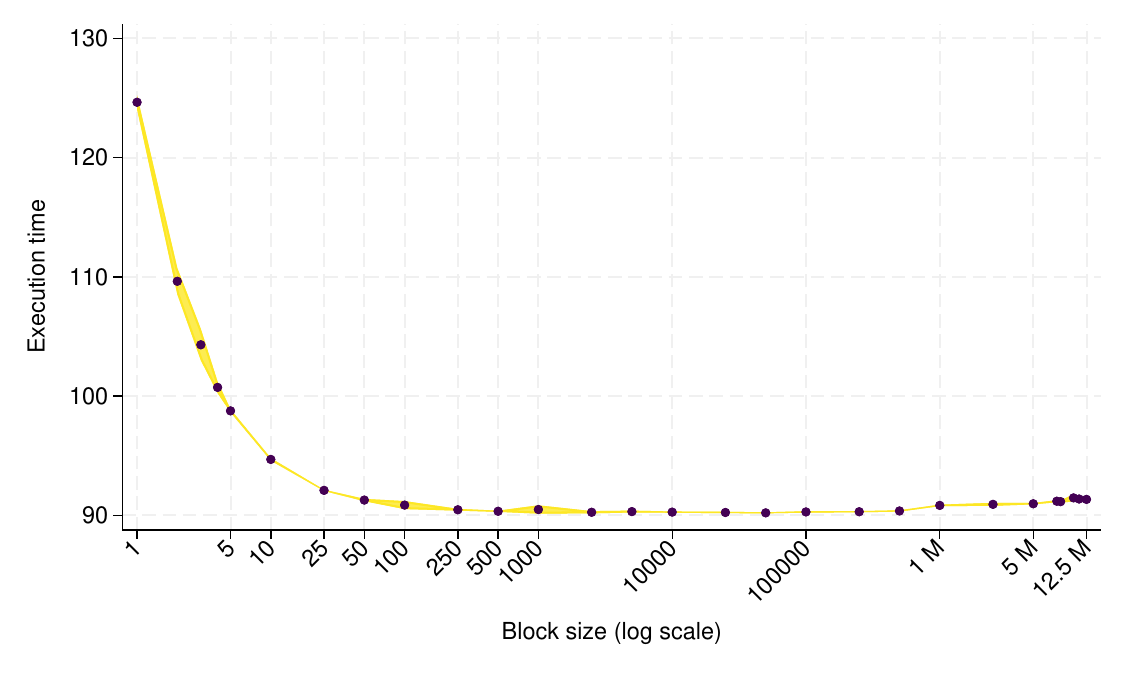}
  \caption{No Binding Memory Limits}
  \end{subfigure}
  \floatfoot{\textbf{Notes}: Total time taken to estimate a regression with 5 independent variables and 25 million observations is documented.  Varying block sizes are used, as plotted on the horizontal axes, and total time is plotted on the vertical axis.  The left-hand panel is estimated on a server with 1GB of RAM, 4 cores, and an Intel i7 processor.  The right-hand panel is estimated on a PC with 32GB of RAM, 8 cores and Intel i7 processors. In each case, estimates are generated 50 times for each block size and average times are plotted as circles.  The 95\% confidence interval of these estimation times are plotted in yellow.}
\end{figure}

\subsection{An Empirical Example}
\label{sscn:empirics}
We document the performance of cumulative algorithms and their non-cumulative counterparts on a real empirical example.  This empirical example is based on a large sample of microdata, following \shortciteN{Aaronsonetal2020}. \shortciteANP{Aaronsonetal2020}\ estimate the impact of fertility on mother's labour supply using data over 2 centuries from censuses and demographic surveys.  We follow \shortciteN{Aaronsonetal2020} in downloading data from IPUMS and the Demographic and Health Surveys resulting in 51,449,770 observations covering 106 countries, with observations drawn from 434 country by year cells.  Data covers years 1787-2015, and measures women's labour force participation, total fertility, and a number of other mother-level covariates.  In Appendix \ref{scn:data} we provide summary statistics as well as a graph documenting the years covered in data and a graph documenting the countries covered and the number of observations in each (Figures \ref{fig:countries} and \ref{fig:years}).

This example is well-suited to our setting because it allows us to document the relative performance of a number of different estimation and inference procedures.  Specifically, two models are considered, and these are estimated in a number of ways.  A first model is simple (weighted) ordinary least squares, where each woman's labour for participation measure is regressed on her total fertility.  We estimate:
\begin{equation}
\label{eqn:AaronsonOLS}
\text{Participation}_{ict} = \beta_0 + \beta_1 \text{Fertility}_{ict} + X_{ict}^\prime\beta + \phi_{c\times t} + \varepsilon_{ict},
\end{equation}
for individual $i$ in country $c$ observed in year $t$, where country by year fixed effects are indicated as $\phi$, and covariates $X_{ict}$ are those indicated by  \shortciteANP{Aaronsonetal2020}; namely each women's age, age at first birth, and first born child's sex.  Second, we estimate an IV model, where in the first stage, a measure of fertility (specifically whether a woman has a third child) is regressed on an indicator of a woman having second birth twins, and then in a second stage labour force participation is regressed on instrumented fertility:
\begin{eqnarray}
\label{eqn:AaronsonIV}
\text{Fertility 3}_{ict}&=&\pi_0 + \pi_1 \text{Twin 2}_{ict} + X_{ict}^\prime\Pi + \phi_{c\times t}+\nu_{ict} \nonumber \\
\text{Participation}_{ict} &=& \gamma_0 + \gamma_1 \widehat{\text{Fertility 3}}_{ict} + X_{ict}^\prime\Gamma + \phi_{c\times t} + \eta_{ict}.
\end{eqnarray}
All other details follow those laid out in \eqref{eqn:AaronsonOLS}, and replicate models proposed by \shortciteN{Aaronsonetal2020}.\footnote{All results are replicated exactly.}  This IV strategy follows a long tradition, starting with \citeN{RosenzweigWolpin1980}, of seeking to draw conditionally exogenous variation in fertility owing twin births (see \citeN{BhalotraClarke2023} for a recent overview).

In this context, we are interested in documenting the processing times of IV and OLS estimation of coefficients and standard errors with a number of specific estimation procedures.  This includes IV and OLS models where fixed effects are directly estimated as well as estimation by fixed effects estimators where fixed effects are concentrated out resulting in estimates only of coefficients on time-varying variables. We also consider a number of alternative inference procedures; namely, firstly assuming homoscedasticity, then clustering standard errors by country$\times$year, both analytically, and with a clustered bootstrap.  We report the processing time of cumulative algorithms written by us to implement the procedures we lay out above compared to commercially produced (non-cumulative) algorithms written in Stata (version 18). We also consider alternative non-commercial (non-cumulative) algorithms which implement potentially more efficient fixed effect procedures, namely a more rapid implementation of the within transformation described by \citeN{Gaure2013,GuimaraesPortugal10,Correia2016}, implemented by \citeN{Correia2016}. All processing times are measured in minutes and include the time of reading data, estimating the regression and producing output, and are estimated under controlled conditions on a a server with fixed characteristics which are varied across tests.  Each procedure is estimated 10 times, with average processing times reported.

\begin{table}[ht!]
\caption{\small Execution Time Across Models in a Labour Market Participation Example: OLS and IV}
\label{tab:processingTimesSE}
\begin{tabular}{lcccccc}\toprule
& \multicolumn{5}{c}{Memory Constraints} \\ \cmidrule(r){2-6}
& 1 GB & 2 GB & 8 GB & 16 GB & 32 GB \\ \midrule
\multicolumn{6}{l}{\underline{Panel A: OLS}}\\
\multicolumn{6}{l}{\textbf{Cumulative Procedures}}\\
\ \ \ \ Within-transformation&4.060&4.135&4.111&4.236&4.108\\
\ \ \ \ Within, Clustered&7.979&8.053&7.994&8.174&7.929\\
\ \ \ \ Within, Clustered Bootstrap&4.086&4.133&4.151&4.209&4.084\\
\ \ \ \ Clustered Bootstrap (s)&3.684&3.745&3.798&3.815&3.715\\
\ \ \ \ Full Estimation&4.016&4.010&4.046&4.109&3.951\\
\multicolumn{6}{l}{\textbf{Standard Procedures}}\\
Within-transformation (areg)&\cellcolor{lightgray}--&13.677&8.669&8.290&8.497\\
\ \ \ \ Within, Clustered&\cellcolor{lightgray}--&\cellcolor{lightgray}--&12.034&11.029&11.333\\
\ \ \ \ Full Estimation&\cellcolor{lightgray}--&16.474&10.147&10.103&10.054\\
\ \ \ \ Within-transformation (hdfe)&\cellcolor{lightgray} --&\cellcolor{lightgray} --&\cellcolor{lightgray} --&3.605&3.476\\
\ \ \ \ Within, Clustered (hdfe)&\cellcolor{lightgray} --&\cellcolor{lightgray} --&\cellcolor{lightgray} --&3.794&3.554\\\midrule
\multicolumn{6}{l}{\textbf{Panel B: IV}}\\
\multicolumn{6}{l}{\textbf{Cumulative Procedures}}\\
\ \ \ \ Within-transformation&4.600&4.608&4.609&4.698&4.589\\
\ \ \ \ Within, Clustered&8.575&8.653&8.724&8.825&8.636\\
\ \ \ \ Within, Clustered Bootstrap&4.428&4.449&4.429&4.539&4.383\\
\ \ \ \ Winthin, Clustered Bootstrap (s) &4.267&4.293&4.278&4.376&4.271\\
\ \ \ \ Full Estimation &4.362&4.379&4.350&4.448&4.329\\
\multicolumn{6}{l}{\textbf{Standard Procedures}}\\
\ \ \ \ Within-transformation (xtivreg)&\cellcolor{lightgray}--&\cellcolor{lightgray}--&\cellcolor{lightgray}--&\cellcolor{lightgray}--&6.629\\
\ \ \ \ Within, Clustered&\cellcolor{lightgray}--&\cellcolor{lightgray}--&\cellcolor{lightgray}--&\cellcolor{lightgray}--&15.958\\
\ \ \ \ Full Estimation&\cellcolor{lightgray}--&\cellcolor{lightgray}--&\cellcolor{lightgray}--&\cellcolor{lightgray}--&226.181\\
\ \ \ \ Within-transformation (hdfe)&\cellcolor{lightgray} --&\cellcolor{lightgray} --&\cellcolor{lightgray} --&\cellcolor{lightgray} --&5.996\\
\ \ \ \ Within, Clustered (hdfe)&\cellcolor{lightgray} --&\cellcolor{lightgray} --&\cellcolor{lightgray} --&\cellcolor{lightgray} --&8.597\\ \midrule
\textbf{Panel C: System Benchmark}&768.1&764.9&769.9&772.7&778.5\\\bottomrule
\end{tabular}

\vspace*{0.1cm}

\begin{minipage}{13.8cm}
\footnotesize \textbf{Notes}: Each cell reports average processing time in minutes of a particular estimation procedure based on the memory constraints listed in column headers.  All averages are taken over 10 estimations.  Cells are coloured gray when estimation cannot occur due to memory limits.  All methods are estimated in Stata 18 SE.  For clustered bootstrap errors, (s) stands for ``sorted'', implying that the database is sorted by clusters prior to processing. System benchmark is a standard test of processor capacity run on the operating system to provide a baseline comparison of server performance across columns.
\end{minipage}
\end{table}

Results are displayed in Table \ref{tab:processingTimesSE}.  Each cell provides the average processing time for a particular estimation procedure in a particular computational environment.  Estimation procedures are listed in rows, and computational resources are listed in columns. In Panel A we document times corresponding to OLS \eqref{eqn:AaronsonOLS}, and in Panel B we document times corresponding to IV \eqref{eqn:AaronsonIV}.  Within panels, we first present processing times for cumulative algorithms, and below this, standard regression or IV regression implementations.  We replicate each process for a range of computational systems.  These are all commercially available dedicated private servers with virtual memory limits (RAM) listed in column headers.  Although each column principally changes the quantity of RAM available on the server, a number of other more minor changes may occur on the server when changing configurations.  For this reason, in Panel C we provide a system benchmark which shows the server's performance on a standard numerical test.\footnote{Specifically, the system benchmark consists of counting the number of times that the computer is capable of calculating all the primes up to 10,000 in a 10 second span.  In this case, higher values of the system benchmark imply that the computer is faster.}

Considering the behaviour of cumulative algorithms in OLS, we see that irrespective of the computer's memory, the processing time is very similar, ranging from an average of 4.0 to 4.2 minutes when estimation is conducted by within-transformed variables as in \eqref{eqn:FE}.  This stability across systems is precisely the value of cumulative regression procedures.  Here, whether one has available a system with substantial memory (32 GB) or very little memory (1 GB), there is no change in performance.  The case of standard regression is of course different.  In systems with small memory capacity it is simply infeasible to load data and estimate parameters.  These cells are shaded in gray. Initially, when loading data into a system with 2 GB of memory is viable, the standard implementation proves markedly slower than its cumulative counterpart (13.67 minutes compared to 4.1 minutes). Yet, this instance is particularly noteworthy, revealing an overflow in the standard procedure at the threshold of memory usage limits.  Moving to larger memory capacities reduces the time of processing of within transformed models slightly (to around 8.3 to 8.6 minutes), but given the efficiency with which cumulative routines can conduct fixed effect procedures (Section \ref{sscn:groups}), standard implementations do not approach the performance of cumulative algorithms.  Alternative routines for dealing with fixed effects are observed to be slightly faster when they are feasible to be estimated, as observed when within-transformed models from cumulative algorithms are compared with hdfe implementations.  However, such models are infeasible in a range of cases where lower memory limits are binding, and similarly cannot return full parameter vectors.

When we wish to report full parameter vectors (``Full Estimation''), cumulative algorithms outperform comparison estimators.  Across system types, cumulative procedures require between 3.9 to 4.2 minutes for estimation, while similar procedures in non-cumulative models require around 10 minutes.  It is noteworthy that in the case of cumulative algorithms, full estimation including fixed effects is marginally faster than within-transformed versions, and this owes to the fact that within transformations require conducting group-level analyses to store matrices $\Sigma$ and $\Upsilon$ for each group, while models to return full fixed effects do not.

Turning to inference, we observe that when calculating clustered standard errors analytically, the processing time of cumulative algorithms approximately doubles.  This owes to the procedure laid out in Section \ref{sscn:inference} where the calculation of standard errors requires estimates of $\widehat{u}$, and hence requires reading data two times.  However, one particular advantage of fixed effect estimation is that if clustered standard errors are desired, running block bootstraps is virtually costless.  In Table \ref{tab:processingTimesSE} we run 500 bootstraps following the procedure laid out in Section \ref{sscn:bootstrap}, seeing that this adds nearly no processing time (row 1 compared to row 3 of Table \ref{tab:processingTimesSE}).  We do not display the comparison for standard clustered bootstrap processing times in the table in a \emph{non-cumulative} process, simply because this procedure increases linearly in the number of bootstraps, and is many orders of magnitude slower than cumulative procedures.  In general, all group-level processing in Table \ref{tab:processingTimesSE} occurs when data is sorted in any order, however if data is indeed sorted by group prior to processing, estimation time further falls given that group-level processing of data can occur more rapidly (row 4 of panel A).

In the case of IV models, results are even more stark given that non-cumulative calculations require housing larger matrices in memory, and make processing infeasible at a broader range of memory restrictions.  Indeed, in this case, even with data with only approximately 50 million observations, processing in non-cumulative routines becomes feasible at 32 GB of memory, but not before.  We again see a number of key take-aways from Panel B of Table \ref{tab:processingTimesSE}.  These are, firstly, that cumulative routines open up feasibility where previously estimation could not occur.  Secondly, even where processing can occur, cumulative algorithms are generally faster than non-cumulative procedures.  In some cases, presumably where data approaches the memory limits of the computer, comparisons suggest that even where feasible, non-cumulative methods may be considerably slower than cumulative counterparts.  This is especially clear in the comparison of IV models where all fixed effects are estimated (226 minutes) to a similar procedure in cumulative 2SLS (4.3 minutes).

It is important to note that this example should be conceived as simply illustrative of the fact that along with the conceptual interest of cumulative procedures, they do have practical implications too.  Our implementations are unlikely to compare to commercial implementations in terms of the efficiency of every internal process, so could be conceived as lower bounds of the performance in this particular setting.  In other languages and on other specific computational environments results will also vary.  If instead of conducting these tests in single threaded versions of Stata we conduct them in multiple processor environments, results are observed to be similar in nature (Appendix Table \ref{tab:processingTimesMP}). 
What's more, while this data has a reasonable number of observations (around 50 million), estimation is feasible with as little as 8GB of memory depending on the estimation procedure.  However, in cases with much larger data, such memory limits will be far more binding, suggesting that the feasibility benefit of cumulative algorithms may be more important.


%

\section{Discussion and Conclusions}
\label{scn:conclusion}
In this paper we show that regressions can be estimated row-wise, without ever requiring all of the information from a dataset at a given moment of time.  In the limit, we show that regressions can be estimated in a simple fashion where only a single line of data is read at a time and then forgotten, with only a small number of low dimensional matrices needing to be updated over time.  This fact appears to have been documented in very early computational work in economics \cite{Brownetal1953}, however only for a very specific variant of the problem.

Despite the ubiquity of procedures which work in a column-wise fashion in econometrics---based on results known since the seminal papers of \citeN{FrischWaugh33,Lovell63}---to our knowledge these row-wise results have received very scarce attention.  We show that these results hold for a broad class of regression models including OLS, IV, fixed effect, and regularised regression models such as Ridge and LASSO, and that the logic of these results holds both for both least squares and other M-class estimators such as probit and logit models.  These results are not approximations, providing exact calculations of regression coefficients, and can similarly generate standard errors and other regression statistics such as goodness of fit measures exactly.  Turning to inference, we show that these methods apply for both homoescedastic error assumptions, as well as for heteroscedasticity- and cluster-robust variance estimators.  We additionally document that certain bootstrap procedures and the definition of tuning parameters can be substantially more efficient with applications of the results from this paper.

As well as the theoretical interest of understanding the mechanics of frequently used regression models, these results have a large number of practical uses.  In both simulated and real data we show that these results imply that models which cannot be estimated on certain computers using standard commercial implementations of regression software can be estimated using the same programs but with our algorithms.  What's more, even in cases where a computer's memory does not limit data from being opened, we show that our algorithms can at times offer non-trivial speed ups over standard software.  In some sense these results could be cast as democratising the processing of big data via regression, as, provided that a sufficiently large hard disk is available, one could process extremely large datasets with very low memory requirements.  Indeed, there is no reason why these results could not be applied to data stored remotely, or on the web, implying that it would not be necessary to have access to super computers to process data of any size.

The results in this paper are likely lower bounds of the true performance of these algorithms.  We have implemented the results in this paper in a high level matrix processing language, and comparisons are made to programs written largely in faster low-level languages.  What's more, the procedures we use here process all data in a naive sequential fashion.  The results from this study make clear that regression is an embarrassingly parallelisable task, and so if data is stored or broken down into various chunks in different files, processing times likely also scale approximately linearly in parallel processes.  All told, these results suggest that transposing the ideas of Frisch-Waugh-Lovell to process regressions in a row-by-row rather than column-by-column fashion provides both an interesting theoretical proposition, as well as useful practical applications.
\end{spacing}
\newpage
\bibliographystyle{chicago}
\bibliography{refs}

\setcounter{table}{0}
\renewcommand{\thetable}{A\arabic{table}}

\setcounter{figure}{0}
\renewcommand{\thefigure}{A\arabic{figure}}

\clearpage
\appendix
\begin{center}
\textbf{Appendices for: Frisch-Waugh-Lovell$^\prime$} \\
Damian Clarke, Nicol\'as Paris Torres \& Benjam\'in Villena-Rold\'an \\
Not for print.
\end{center}

\newpage
\begin{landscape}
\section{A Simple Visualisation in Matrix Form}
\label{app:simpleVis}
Consider a simple illustration of the generation of $X^\prime X$ based on a case where $X$ is a matrix consisting of $N=4$ observations and $K=3$ independent variables.  For ease of visualisation, we will denote $X$ as follows:
\begin{equation}
\nonumber
 \left(
  \begin{array}{ccc}
    \textcolor{blue}{x_{1,1}}& \textcolor{blue}{x_{1,2}} & \textcolor{blue}{x_{1,3}}  \\
    x_{2,1}& x_{2,2} & x_{2,3}  \\
    x_{3,1}& x_{3,2} & x_{3,3}  \\
    \textcolor{red}{x_{4,1}}& \textcolor{red}{x_{4,2}} & \textcolor{red}{x_{4,3}}  \\
  \end{array}\right) 
\end{equation}
where all realisations of independent variables for observation $i=1$ are coloured in blue, and for observation $i=4$ are coloured in red.  Now note that $X^\prime X$ can be written in extensive form as below:
\begin{eqnarray}
 X^\prime X &\equiv&
 \left(
  \begin{array}{ccccc}
    \textcolor{blue}{x_{1,1}}& x_{2,1} & x_{3,1} & \textcolor{red}{x_{4,1}}\\
    \textcolor{blue}{x_{1,2}}& x_{2,2} & x_{3,2} & \textcolor{red}{x_{4,2}}\\
    \textcolor{blue}{x_{1,3}}& x_{2,3} & x_{3,3} & \textcolor{red}{x_{4,3}}\\
  \end{array}\right) 
 \left(
  \begin{array}{ccc}
    \textcolor{blue}{x_{1,1}}& \textcolor{blue}{x_{1,2}} & \textcolor{blue}{x_{1,3}}  \\
    x_{2,1}& x_{2,2} & x_{2,3}  \\
    x_{3,1}& x_{3,2} & x_{3,3}  \\
    \textcolor{red}{x_{4,1}}& \textcolor{red}{x_{4,2}} & \textcolor{red}{x_{4,3}}  \\
  \end{array}\right)  \nonumber \\
  &=& \left( 
  \begin{array}{ccc}
     \textcolor{blue}{x^2_{1,1}}+x^2_{2,1}+x^2_{3,1}+\textcolor{red}{x^2_{4,1}} & \textcolor{blue}{x_{1,1}}\textcolor{blue}{x_{1,2}} + x_{2,1}x_{2,2} + x_{3,1}x_{3,2} + \textcolor{red}{x_{4,1}x_{4,2}} & \textcolor{blue}{x_{1,1}}\textcolor{blue}{x_{1,3}} + x_{2,1}x_{2,3} + x_{3,1}x_{3,3} + \textcolor{red}{x_{4,1}x_{4,3}}  \\ 
    \textcolor{blue}{x_{1,2}}\textcolor{blue}{x_{1,1}} + x_{2,2}x_{2,1} + x_{3,2}x_{3,1} + \textcolor{red}{x_{4,2}x_{4,1}} & 
     \textcolor{blue}{x^2_{1,2}}+x^2_{2,2}+x^2_{3,2}+\textcolor{red}{x^2_{4,2}} & \textcolor{blue}{x_{1,3}}\textcolor{blue}{x_{1,1}} + x_{2,3}x_{2,1} + x_{3,3}x_{3,1} + \textcolor{red}{x_{4,3}x_{4,1}}  \\ 
    \textcolor{blue}{x_{1,3}}\textcolor{blue}{x_{1,1}} + x_{2,3}x_{2,1} + x_{3,3}x_{3,1} + \textcolor{red}{x_{4,3}x_{4,1}} & 
     \textcolor{blue}{x_{1,3}}\textcolor{blue}{x_{1,2}} + x_{2,3}x_{2,2} + x_{3,3}x_{3,2} + \textcolor{red}{x_{4,3}x_{4,2}} & 
      \textcolor{blue}{x^2_{1,3}}+x^2_{2,3}+x^2_{3,3}+\textcolor{red}{x^2_{4,3}}\\ 
  \end{array}
  \right)\nonumber \\
  &=& \left( 
  \begin{array}{ccc}
    \textcolor{blue}{x^2_{1,1}} & \textcolor{blue}{x_{1,1}}\textcolor{blue}{x_{1,2}} & \textcolor{blue}{x_{1,1}}\textcolor{blue}{x_{1,3}} \\
     \textcolor{blue}{x_{1,2}}\textcolor{blue}{x_{1,1}}  & \textcolor{blue}{x_{1,3}}\textcolor{blue}{x_{1,1}} & \textcolor{blue}{x_{1,3}}\textcolor{blue}{x_{1,1}} \\
      \textcolor{blue}{x_{1,3}}\textcolor{blue}{x_{1,1}} &  \textcolor{blue}{x_{1,3}}\textcolor{blue}{x_{1,2}} &   \textcolor{blue}{x^2_{1,3}} \\
  \end{array}
  \right)
  +
  \left( 
  \begin{array}{ccc}
  x^2_{2,1} & x_{2,1}x_{2,2} & x_{2,1}x_{2,3} \\
  x_{2,2}x_{2,1} & x^2_{2,2} & x_{2,3}x_{2,1} \\
  x_{2,3}x_{2,1} & x_{2,3}x_{2,2} & x^2_{2,3} \\
  \end{array}
  \right)
  +
  \cdots
  +
\left( 
  \begin{array}{ccc}
  \textcolor{red}{x^2_{4,1}} & \textcolor{red}{x_{4,1}x_{4,2}} & \textcolor{red}{x_{4,1}x_{4,3}}  \\
  \textcolor{red}{x_{4,2}x_{4,1}} & \textcolor{red}{x^2_{4,2}} &\textcolor{red} {x_{4,3}x_{4,1}} \\
  \textcolor{red}{x_{4,3}x_{4,1}}  & \textcolor{red}{x_{4,3}x_{4,2}}  & \textcolor{red}{x^2_{4,3}} \\
  \end{array}
  \right). \label{eqn:extensiveMat} 
\end{eqnarray}
The key takeaway here is that one can simply arrive to the $K\times K$ matrix $X^\prime X$ by calculating 4 $K\times K$ matrices based on cross-products for each observation (ie working an observation at a time), and finally summing across these matrices (as in (\ref{eqn:extensiveMat})).

\end{landscape}

\section{Appendix Figures and Tables}
\begin{figure}[htpb!]
  \caption{Sample size and execution time of standard versus cumulative (Multiple Cores)} 
  \label{tab:processingTimesMP}
  \begin{subfigure}{0.49\textwidth}
    \includegraphics[scale=0.52]{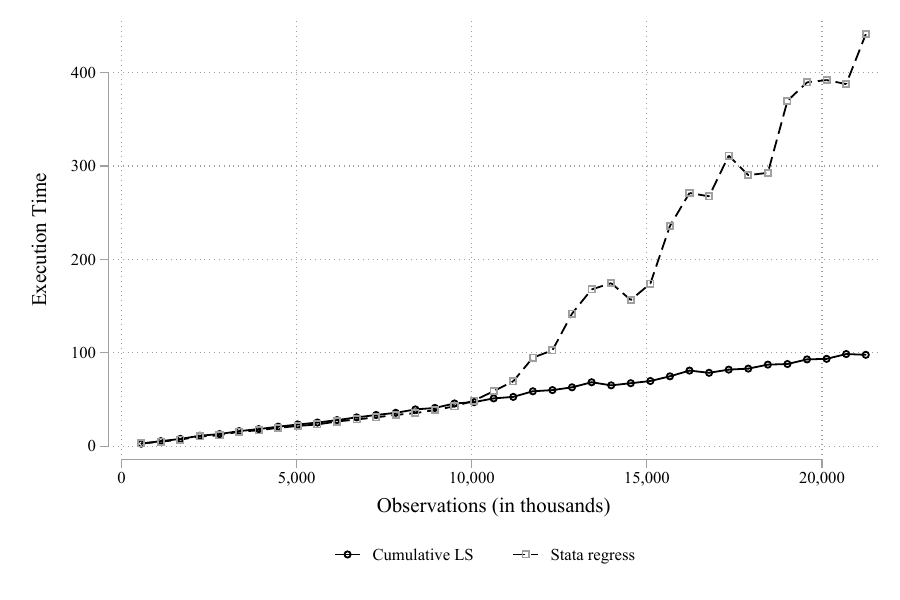}
  \caption{5 Independent Variables}
  \end{subfigure}
  \begin{subfigure}{0.49\textwidth}
    \includegraphics[scale=0.52]{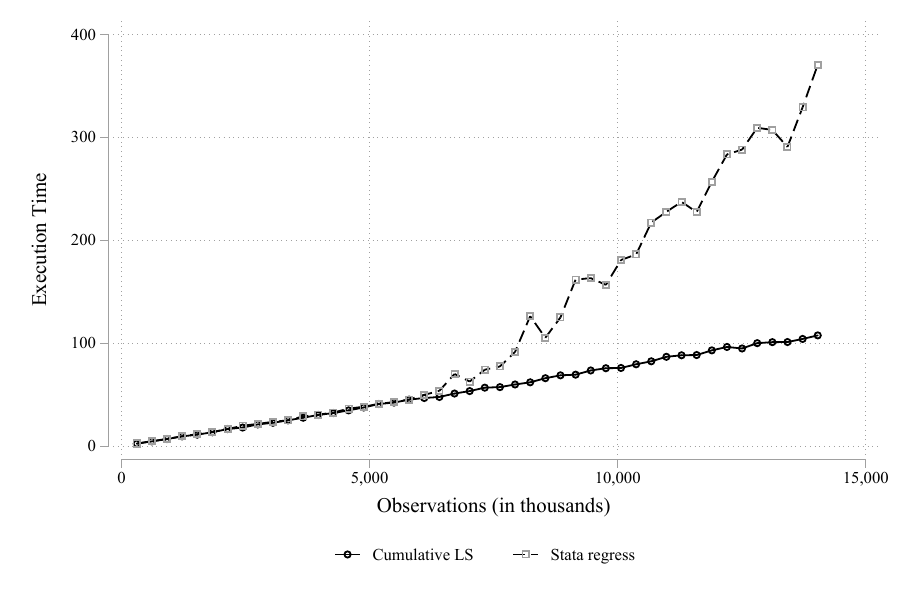}
  \caption{10 Independent Variables}
  \end{subfigure}

  \begin{subfigure}{0.49\textwidth}
    \includegraphics[scale=0.52]{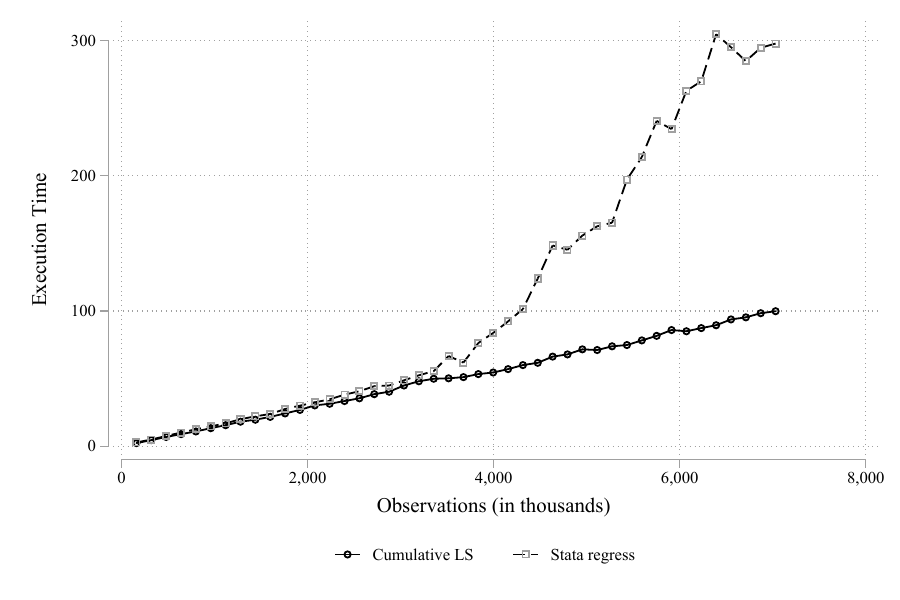}
  \caption{20 Independent Variables}
  \end{subfigure}
  \begin{subfigure}{0.49\textwidth}
    \includegraphics[scale=0.52]{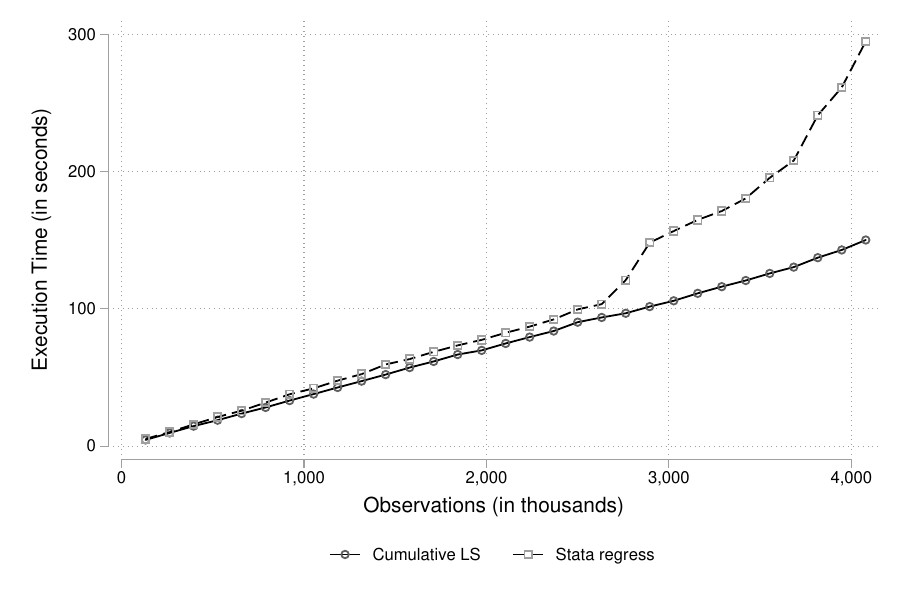}
  \caption{50 Independent Variables}
  \end{subfigure}
  \floatfoot{\textbf{Notes}: All times refer to the computation time of reading data into memory and estimating an ordinary least squares regression. Tests are all conducted on a system with 1GB of RAM, with no other processes running, using Stata MP (2 cores).  In each case, tests are conducted up to the point at which there is insufficient RAM to open the data, thus precluding the estimation of standard regression models. Beyond this point, it is still feasible to estimate parameters using Cumulative Least Squares.}
\end{figure}


\begin{figure}
    \centering
    \caption{Relative Performance of OLS to CLS by Block Size}
    \label{fig:block3d}
    \begin{subfigure}{0.49\textwidth}
    \includegraphics[scale=0.52]{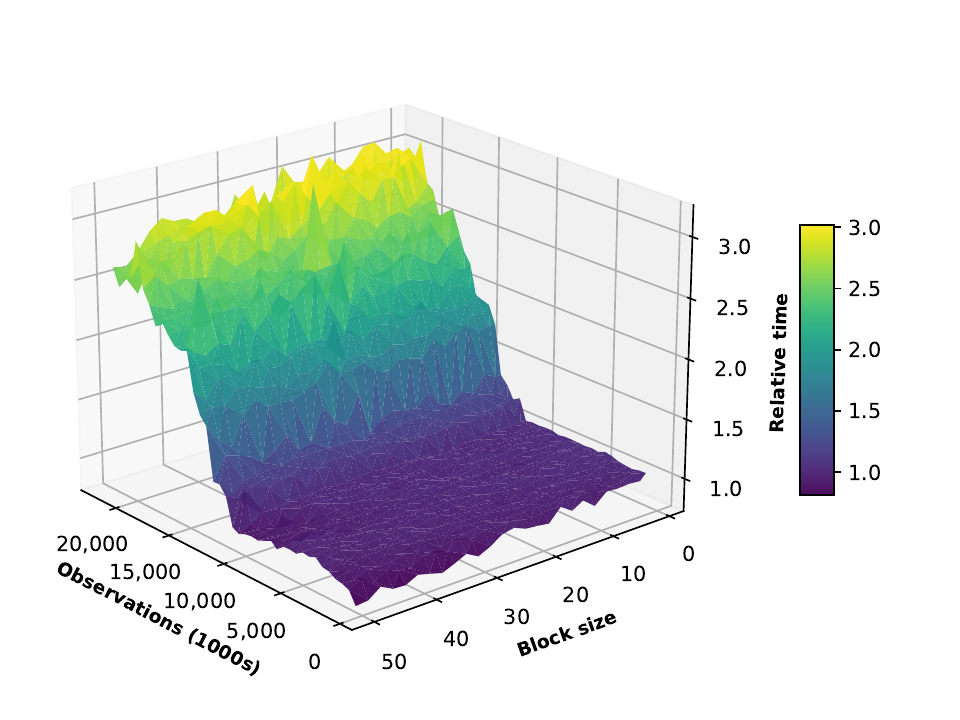}
    \caption{5 Independent Variables}
    \end{subfigure}
    \begin{subfigure}{0.49\textwidth}
    \includegraphics[scale=0.52]{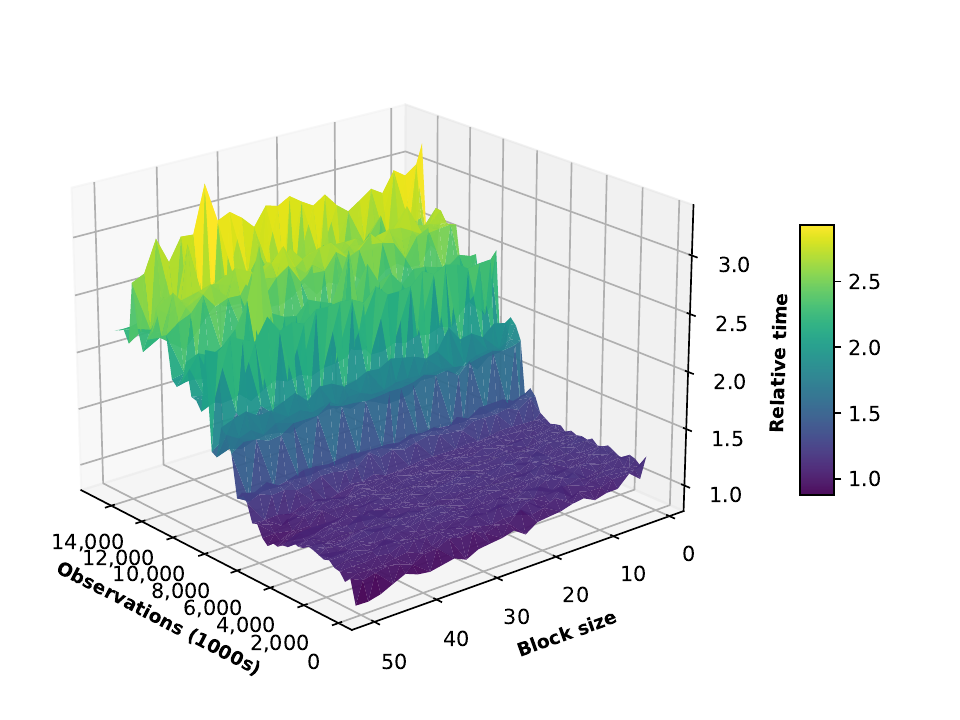}
    \caption{10 Independent Variables}
    \end{subfigure}

    \begin{subfigure}{0.49\textwidth}
    \includegraphics[scale=0.52]{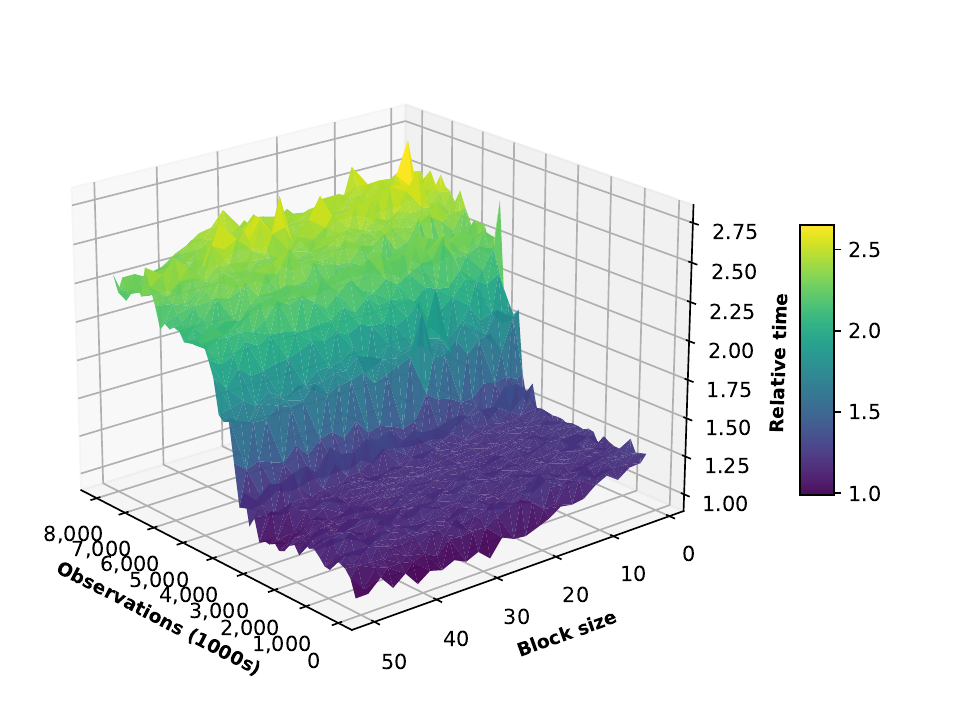}
    \caption{20 Independent Variables}
    \end{subfigure}
    \begin{subfigure}{0.49\textwidth}
    \includegraphics[scale=0.52]{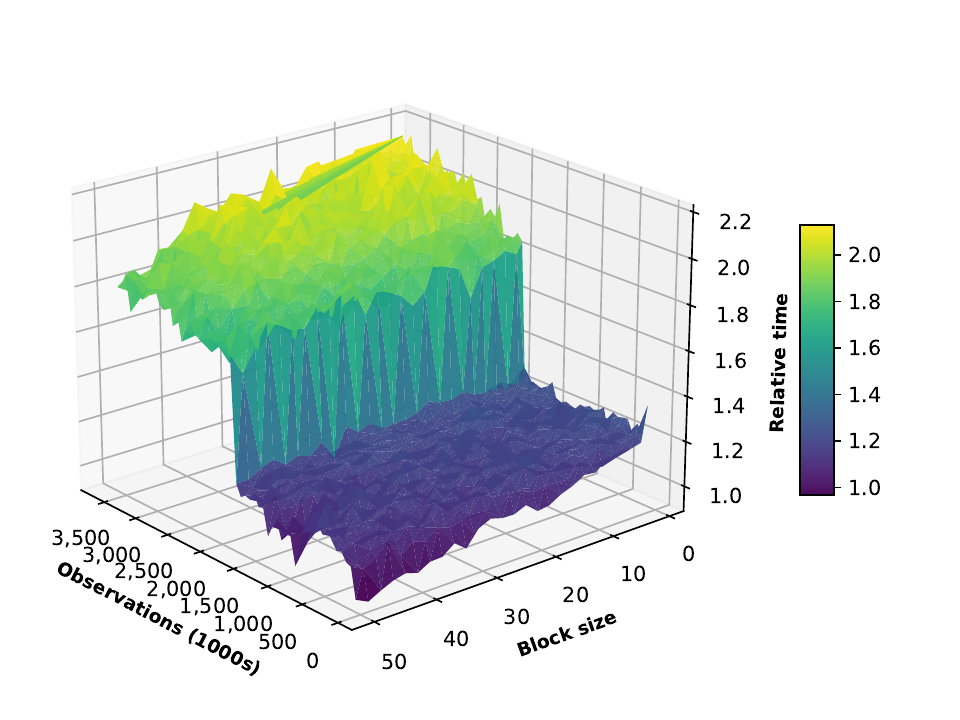}
    \caption{50 Independent Variables}
    \end{subfigure}
    \floatfoot{\textbf{Notes}: Surface plots present the ratio of computation times of Stata's native regression command to cumulative least squares.  Values less than 1 imply non-cumulative procedures are faster than cumulative procedures, and vice versa for values greater than 1. All times refer to the computational time of reading data into memory and estimating a cumulative least squares regression. Tests are all conducted on a system with 1GB of RAM, with no other processes running.  Block sizes indicated on the horizontal axis refer to the proportion of the full dataset (\emph{eg} 50 refers to 2 blocks each covering 50\% of the data, 1 refers to 100 blocks each covering 1\% of the data).}
\end{figure}

\begin{table}[htpb!]
\caption{\small Execution Time Across Models in a Labour Market Participation Example (MP): OLS and IV}
\begin{tabular}{lcccccc}\toprule
& \multicolumn{5}{c}{Memory Constraints} \\ \cmidrule(r){2-6}
& 1 GB & 2 GB & 8 GB & 16 GB & 32 GB \\ \midrule
\multicolumn{6}{l}{\underline{Panel A: OLS}}\\
\multicolumn{6}{l}{\textbf{Cumulative Procedures}}\\
\ \ \ \ Within-transformation&2.973&3.153&3.266&3.091&3.303\\
\ \ \ \ Within, Clustered&5.735&6.003&6.315&5.973&6.438\\
\ \ \ \ Within, Clustered Bootstrap&3.070&3.002&3.218&3.124&3.342\\
\ \ \ \ Clustered Bootstrap (s)&2.584&2.760&2.930&2.639&3.131\\
\ \ \ \ Full Estimation&2.843&2.895&3.181&2.960&3.185\\
\multicolumn{6}{l}{\textbf{Standard Procedures}}\\
Within-transformation (areg)&\cellcolor{lightgray}--&\cellcolor{lightgray} --&4.511&4.187&3.846\\
\ \ \ \ Within, Clustered&\cellcolor{lightgray}--&\cellcolor{lightgray}--&5.368&5.257&4.532\\
\ \ \ \ Full Estimation&\cellcolor{lightgray}--&\cellcolor{lightgray} --&4.500&4.143&3.630\\
\ \ \ \ Within-transformation (hdfe)&\cellcolor{lightgray} --&\cellcolor{lightgray} --&\cellcolor{lightgray} --&3.397&3.177\\
\ \ \ \ Within, Clustered (hdfe)&\cellcolor{lightgray} --&\cellcolor{lightgray} --&\cellcolor{lightgray} --&3.537&3.304\\\midrule
\multicolumn{6}{l}{\textbf{Panel B: IV}}\\
\multicolumn{6}{l}{\textbf{Cumulative Procedures}}\\
\ \ \ \ Within-transformation&3.371&3.435&3.501&3.316&3.651\\
\ \ \ \ Within, Clustered&6.351&6.498&6.798&6.374&6.950\\
\ \ \ \ Within, Clustered Bootstrap&3.242&3.312&3.479&3.306&3.468\\
\ \ \ \ Winthin, Clustered Bootstrap (s) &3.036&3.077&3.192&3.276&3.585\\
\ \ \ \ Full Estimation &2.957&3.003&3.142&2.981&3.538\\
\multicolumn{6}{l}{\textbf{Standard Procedures}}\\
\ \ \ \ Within-transformation (xtivreg)&\cellcolor{lightgray}--&\cellcolor{lightgray}--&\cellcolor{lightgray}--&\cellcolor{lightgray}--&4.968\\
\ \ \ \ Within, Clustered&\cellcolor{lightgray}--&\cellcolor{lightgray}--&\cellcolor{lightgray}--&\cellcolor{lightgray}--&10.916\\
\ \ \ \ Full Estimation&\cellcolor{lightgray}--&\cellcolor{lightgray}--&\cellcolor{lightgray}--&\cellcolor{lightgray}--&45.350\\
\ \ \ \ Within-transformation (hdfe)&\cellcolor{lightgray} --&\cellcolor{lightgray} --&\cellcolor{lightgray} --&\cellcolor{lightgray} --&4.834\\
\ \ \ \ Within, Clustered (hdfe)&\cellcolor{lightgray} --&\cellcolor{lightgray} --&\cellcolor{lightgray} --&\cellcolor{lightgray} --&6.597\\ \midrule
\textbf{Panel C: System Benchmark}&866.2&852.7&856.4&842.4&843.1\\\bottomrule
\end{tabular}

\vspace*{0.1cm}

\begin{minipage}{13.35cm}
\footnotesize \textbf{Notes}: Each cell reports average processing time in minutes of a particular estimation procedure based on the memory constraints listed in column headers.  All averages are taken over 10 estimations.  Cells are coloured gray when estimation cannot occur due to memory limits.  All methods are estimated in Stata 18 MP.  For clustered bootstrap errors, (s) stands for ``sorted'', implying that the database is sorted by clusters prior to processing. System benchmark is a standard test of processor capacity run on the operating system to provide a baseline comparison of server performance across columns.
\end{minipage}
\end{table}

\clearpage
\section{Data Appendix}
\label{scn:data}
We collate original data from IPUMS and the Demographic and Health Survey (DHS) repository using all census data and DHS waves described in \citeN{Aaronsonetal2020}.  This results in 51,449,770 observations drawn from 434 census files for 106 countries covering years 1787 to 2015. The geographical coverage of the data is described in Figure \ref{fig:countries}, and the temporal coverage is described in Figure \ref{fig:years}.

We follow replication materials of \shortciteN{Aaronsonetal2020} to generate all variables, and replicate their results exactly.  We follow their inclusion criteria of working with women aged 21 to 35 who have at least two children, all of whom are 17 or younger.  As described in \shortciteN{Aaronsonetal2020} families are excluded where information is missing on child gender or mother's age, and mothers are not included in the sample if they live in group quarters or give birth before the age of 15.  Summary statistics of all data following the processes described in \shortciteN{Aaronsonetal2020} are included below in Table \ref{tab:SumStats}.  

\begin{figure}[htpb!]
  \caption{Number of Observations by Country}
  \label{fig:countries}
  \includegraphics[scale=0.65]{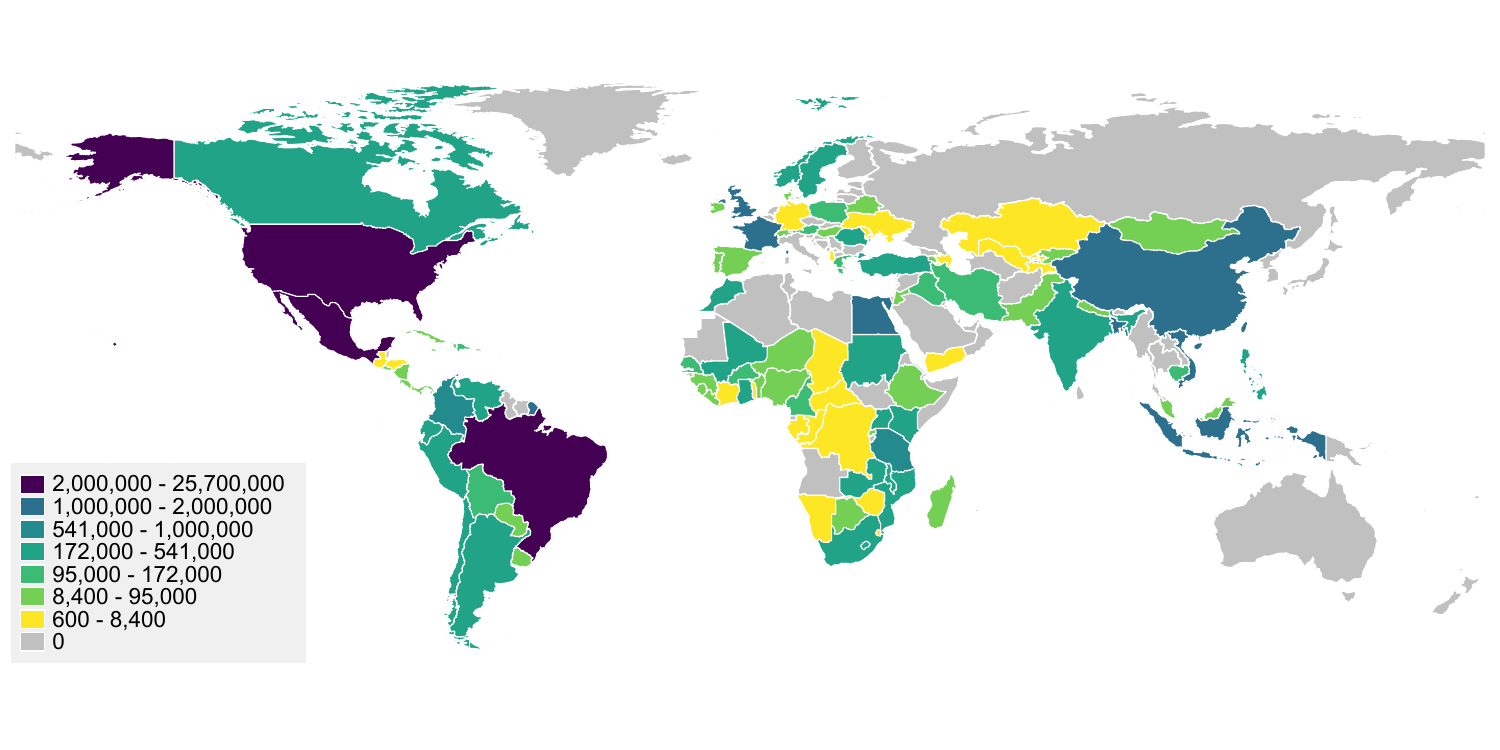}
  \floatfoot{\textbf{Notes:} Values plotted refer to the total frequency of observations used in estimating samples.  These are pooled across all years.  Countries coloured grey have no available microdata on IPUMS or DHS.}
\end{figure}

\begin{figure}[H]
  \caption{Number of Observations by Year}
  \label{fig:years}
  \includegraphics[scale=0.7]{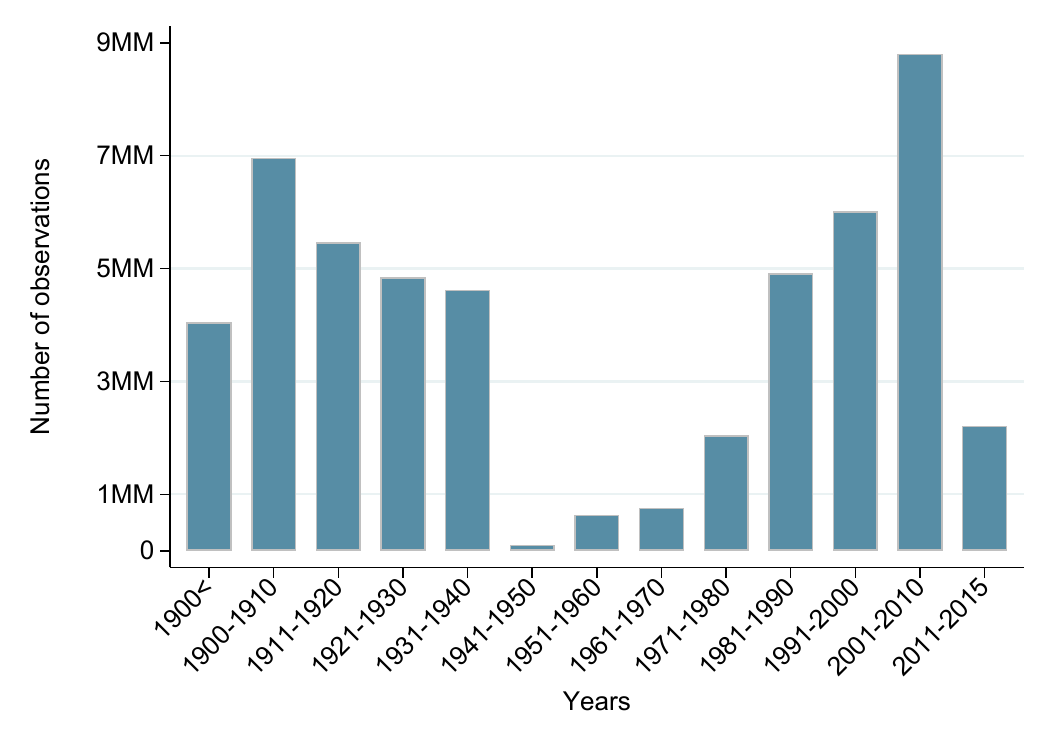}
\begin{minipage}{11cm}
\footnotesize \textbf{Notes:} Values refer to the total frequency of observations used in estimating samples.  Years refer to the year the data was collected.
\end{minipage}
\end{figure}

\begin{table}[htpb!]
\caption{Summary Statistics -- Principal Variables}
\begin{tabular}{lccccc} \toprule
VARIABLES & Obs & Mean & Std. Dev. & Min & Max \\ \midrule
\rule[-0.4cm]{-2.5pt}{0pt}
Labour Force Participation & 51,449,770 & 0.263 & 0.441 & 0 & 1 \\
\textbf{Covariables}  &  & &  & & \\
\cmidrule(l{0.5em}r{1.5em}){1-1}
 \ \  Fertility & 51,449,770 & 0.525 & 0.499 & 0 & 1 \\
\ \  Gender of first child & 51,449,770 & 0.508 & 0.499 & 0 & 1 \\
\ \  Mother Age & 51,449,770 & 29.43 & 3.859 & 21 & 35 \\
\ \  Mother age at first birth & 51,449,770 & 21.04 & 3.300 & 15 & 35 \\
\rule[-0.3cm]{-3pt}{0pt}
\ \  Weights & 51,449,770  & 1 & 0.390 & 0.0007 & 341.66 \\
\rule[-0.4cm]{-2.5pt}{0pt}
\ \  Years & 51,449,770 & 1954.95 & 44.169 & 1787 & 2015 \\ 
\textbf{Instrument}  &  & &  & & \\
\cmidrule(l{0.5em}r{1.5em}){1-1}
\ \  Twin & 51,449,770 & 0.0105 & 0.102 & 0 & 1 \\
 \bottomrule
\end{tabular}

\label{tab:SumStats}
\vspace{0.2cm}

\begin{minipage}{13.2cm}
\footnotesize \textbf{Notes:} Summary statistics are displayed based on all observations and data cleaning procedures described in \shortciteN{Aaronsonetal2020}.  Sample consists of women giving birth at the age of 15 or above, and ages 21 to 35 at the time of data collection.  All selection criteria follows \shortciteN{Aaronsonetal2020}.
\end{minipage}
\end{table}

\clearpage
\section{An Updating Estimation Procedure}
\label{scn:ULS}
Throughout the paper we work with a cumulative procedure in which matrices $X'X$ and $X'y$ are accumulated in a step-wise fashion, and estimates $\widehat\beta_{OLS}$ (or similar for other estimation methodologies) are generated only after data is read in its entirety.  Alternative procedures can be used in which iterations occur over sequential iterations of $\widehat\beta_{OLS}$ estimates themselves, though these are less efficient than the cumulative procedures laid out in the body of the paper. 

To see this, we will use identical notation to that laid out in the paper.  Suppose we wish to estimate a regression between a dependent variable $y$ and a set of $K$ covariates $X_1,X_2,...,X_K$.  The database that can be partitioned into $J$ samples. The whole sample size of the database is $N$, but computing an OLS regression with all the data is unfeasible due to memory constrains. As an alternative procedure, we could run a regression with the sample $j=1$ and update the result with the other $J-1$ samples.  To fix ideas, suppose we have a samples 1 and 2 and we want to compute an OLS estimator for the whole sample denoted by the subindex $1\sim2$. Hence,
\begin{align*}
  \hat{\beta}_{1\sim2} &= \left( \left(
  \begin{array}{cc}
    X_1' & X_2' \\
  \end{array}\right)\left(
  \begin{array}{c}
    X_1 \\
    X_2 \\
  \end{array}
  \right)
  \right)^{-1}  \left(
  \begin{array}{cc}
    X_1' & X_2' \\
  \end{array}\right) \left(
  \begin{array}{c}
    y_1 \\
    y_2 \\
  \end{array}
  \right) = (X_1'X_1 + X_2'X_2)^{-1}(X_1'y_1 +X_2'y_2) \\ &= (\Sigma_1 + \Sigma_2)^{-1}(\Upsilon_1 + \Upsilon_2)
\end{align*}
where $X_j'X_j \equiv \Sigma_j $ and $X_j'y_j \equiv \Upsilon_j$.
The challenge is to compute $\hat{\beta}_{1\sim2}$ using estimates from the two samples separately, so that we can avoid storing very large databases in memory. Trivially, as laid out in the body of the paper, this can be done cumulatively.  Alternatively, we can make use of a result of the inverse of the sum of two matrices by \citeN{Miller81}. A more general perspective in this theory, including application to linear least squares is \citeN{Hager89}. The application of this result can be proven easily.
\begin{lemma}
  The inverse of the sum of two matrices can be obtained as
\begin{align*}
   (\Sigma_1 + \Sigma_2)^{-1} &\equiv \Sigma_{1\sim2}^{-1}= \Sigma_1^{-1} -\Sigma_1^{-1}\Sigma_2 (I +\Sigma_1^{-1} \Sigma_2)^{-1} \Sigma_1^{-1}  \\
\end{align*}
\end{lemma}
\begin{proof} Follows as a direct application of \citeN{Woodbury50} matrix inverse lemma \end{proof}

Using Lemma 1, we can iterate on $\widehat\beta_{OLS}$, without requiring that all of the elements stored in cumulative procedures are accumulated across iterations.  We define the factor 
\[
\Omega_{1\sim 2} \equiv I_K - \Sigma_1^{-1} \Sigma_2 \left(I_K +\Sigma_1^{-1} \Sigma_2\right)^{-1} ,
\]
where $I_K$ stands for an identity matrix of dimension $K$, the number of regressors in the model. Hence, the joint $\Sigma$ matrix can be expressed as
\begin{equation*}
  \Sigma_{1\sim2}^{-1}= \Omega_{1\sim 2}  \Sigma_1^{-1}
\end{equation*}
Therefore, the joint OLS estimator becomes
\begin{equation}
\label{eqn:ULS}
  \widehat{\beta}_{1\sim2} = \Omega_{1\sim 2} \left(\widehat{\beta}_1 + \Sigma_1^{-1} \Upsilon_2\right).
\end{equation}
This suggests an iterative procedure for estimating OLS parameters generalising the above result to $J$ blocks, rather than 2 blocks.  This is defined as Updated Ordinary Least Squares in Algorithm  \ref{alg:ULS}.  
\begin{algorithm}
\caption{Updated Ordinary Least Squares}
\label{alg:ULS}
Inputs: Database consisting of ($y$,$X$), block size $b$. \\
Result: Point estimate $\widehat\beta_{OLS}$ (and potentially point estimates updated at each step, $\widehat\beta_{1\sim j, OLS}$). \\ \ \\
1. Set $i=1$ and $j=b$. Load into memory partition of data covering $y,X$ in observations $i-j$. \;
2. Calculate $\Sigma_1$ and $\widehat\beta_1$ \;
3. If observations $i$-$j$ contain end of file, set $e=1$, otherwise, set $e=0$ \;
\While{$e\neq1$}
{
4. Replace $i=i+b$ and $j=j+b$. Load into memory partition of data covering $y,X$ in observations $i-j$. \;
5. Calculate $\Sigma_j$ and $\Upsilon_j$. \;
6. Calculate $\Omega_{1\sim j}=I_K-\Sigma^{-1}_{j-1}\Sigma_j\left(I_K+\Sigma^{-1}_{j-1}\Sigma_j\right)^{-1}$
and
$\widehat{\beta}_{1\sim j} = \Omega_{1\sim j} \left(\widehat{\beta}_{j-1} + \Sigma_{j-1}^{-1} \Upsilon_j\right)$ \;
7. If observations $i$-$j$ contain end of file, set $e=1$, otherwise, set $e=0$ \;
}\textbf{end}\\ 
\end{algorithm}

While this procedure allows for the calculation of updated values of $\widehat\beta_{1\sim j}$ at each step, and additionally avoids the need of passing $\Upsilon_{1\sim j}$ across steps, each iteration involves one matrix inversion of size $K$ in step 6.  Indeed, for any given quantity of variables $K$ and block size $b$, the cumulative algorithm strictly dominates the updating algorithm in terms of total computations (and hence computation time).  This owes to the fact that the same calculations of order $\mathcal{O}(NK^2)+O(NK)$ discussed in Section \ref{scn:optimal} are required in calculating  $\Sigma_{j}$ and $\Upsilon_j$ as inputs for \eqref{eqn:ULS}, strictly \emph{more} elements are required to be summed in iterating on $\Omega_{1\sim j}$ instead of $\Sigma_{1\sim j}$ and $\Upsilon_{1\sim j}$, 
and additionally, a matrix inversion is required at each step in calculating \eqref{eqn:ULS}.  For this reason, we focus on cumulative least squares algorithms throughout this paper.

\end{document}